\documentclass[a4paper,12pt]{article}

\usepackage{amsthm,amsmath}
\usepackage[round]{natbib}
\usepackage[colorlinks,citecolor=blue,urlcolor=blue]{hyperref}
\usepackage{amssymb}
\usepackage{graphicx}
\usepackage{amsfonts}
\usepackage{amsbsy}
\usepackage{mathrsfs}
\usepackage{lscape}
\usepackage{float}
\usepackage{multirow}
\usepackage{tablefootnote}
\usepackage{threeparttable}
\usepackage{pdflscape}
\usepackage{enumerate}
\usepackage{pdfpages}
\usepackage{authblk}

\theoremstyle{plain}

\DeclareMathOperator{\sign}{sign}

\newtheorem{theorem}{Theorem}
\newtheorem{lemma}{Lemma}

\theoremstyle{definition}

\newtheorem{remark}{Remark}

\graphicspath{{.}{./Figure/}}

\begin{document}


\title{Spatial point processes intensity estimation \\ with a diverging number of covariates}
\author[1]{Achmad Choiruddin}
\author[2, 3]{Jean-Fran\c cois Coeurjolly}
\author[3]{Fr\'ed\'erique Letu\'e}
\affil[1]{Department of Mathematical Sciences, Aalborg University, Denmark}
\affil[2]{Department of Mathematics, Universit\'e du Qu\'ebec \`a Montr\'eal (UQAM), Canada}
\affil[3]{Laboratory Jean Kuntzmann, Department of Probability and Statistics, Universit\'e Grenoble Alpes, France}
\maketitle


\begin{quotation}
\noindent {\it Abstract:} Feature selection procedures for spatial point processes parametric intensity estimation have been recently developed since more and more applications involve a large number of covariates. In this paper, we investigate the setting where the number of covariates diverges as the domain of observation increases. In particular, we consider estimating equations based on Campbell theorems derived from Poisson and logistic regression likelihoods regularized by a general penalty function. We prove that, under some conditions, the consistency, the sparsity, and the asymptotic normality are valid for such a setting. We support the theoretical results by numerical ones obtained from simulation experiments and an application to forestry datasets.

\vspace{9pt}
\par
\end{quotation}\par

\def\thefigure{\arabic{figure}}
\def\thetable{\arabic{table}}

\renewcommand{\theequation}{\thesection.\arabic{equation}}
\fontsize{12}{14pt plus.8pt minus .6pt}\selectfont

\section{Introduction} \label{ch2:intro}
\subsection{Background}
Spatial point pattern data arise in many contexts, e.g. ecology \citep[see e.g.][]{moller2003statistical, renner2015point}, epidemiology \citep[e.g.][]{diggle1990point, diggle2013statistical}, criminology \citep[e.g.][]{baddeley2015spatial, shirota2017statistical}, biology \citep[e.g.][]{illian2008statistical} and astronomy \citep[e.g.][]{baddeley2015spatial}, where interest lies in describing the distribution of an event in space.  Stochastic models generating spatial point patterns are called spatial point processes \citep[see e.g.][]{moller2003statistical, illian2008statistical, diggle2013statistical, baddeley2015spatial}.
	
Usually, the first step to analyze spatial point pattern data is to investigate the intensity. The intensity serves as the first-order characteristics of a spatial point process and often becomes the main interest in many studies, especially when the intensity is suspected to depend on spatial covariates. Examples include the study of spatial variation of specific disease risk related to pollution sources \citep[e.g.][]{diggle1990point, diggle2013statistical}, crime rate analysis in a city related to some demographical information \citep[e.g.][]{shirota2017statistical}, and modeling of the spatial distribution of trees species in a forest related to some environmental factors \citep[e.g.][]{waagepetersen2007estimating, thurman2015regularized, renner2015point}.

We focus in this study on the log-linear model for the intensity function of an inhomogeneous spatial point process defined by
\begin{align}
\label{ch2:intensity function}
\rho (u;\boldsymbol \beta)=\exp(\mathbf{z}(u)^\top \boldsymbol \beta), u \in D \subset \mathbb{R}^d,
\end{align}
where $\mathbf{z}(u)=\{ z_1(u),\ldots,z_p(u)\}^\top$ are the $p$ spatial covariates measured at location $u$, $d$ represents the state space of the spatial point processes (usually $d=2,3$) and $\boldsymbol \beta=\{\beta_1,\ldots,\beta_p\}^\top$ is a real $p$-dimensional parameter. Hence, our main concern is to assess the magnitudes of the vector $\boldsymbol \beta$. For parametric estimation, while maximum likelihood estimation \citep[e.g.][]{berman1992approximating, rathbun1994asymptotic} has been widely implemented for Poisson point processes models, estimating equation-based methods \citep[e.g.][]{waagepetersen2007estimating,  waagepetersen2008estimating, guan2010weighted, baddeley2014logistic} are simpler to implement for more general spatial point processes models, overcoming the possible drawback of MCMC methods which are usually computational expensive \citep{moller2003statistical}. However, when the number of covariates is relatively large, maximum likelihood estimation  and estimating equation-based methods become undesirable: all covariates are selected yielding an increasing standard error for parameter estimates.

\subsection{Feature selection techniques}

To select significant covariates, one may consider a traditional procedure such as a stepwise method. This technique starts with an initial set of covariates, then considers adding or deleting a covariate from the current set at each iteration using a criterion such as an F-statistic or AIC. However, such procedure has a number of limitations: it can be numerically unstable and exhibits high variance due to its discrete procedure \citep[e.g.][]{breiman1996heuristics,fan2001variable,friedman2008elements}. It is even computationally unfeasible  especially when the number of covariates is too large \citep[e.g.][]{breiman1996heuristics,zou2006adaptive}. 
	
To overcome this drawback, regularization techniques have recently been developed for spatial point processes intensity estimation. Such methods are able to perform variable selection while keeping interesting properties in terms of prediction. For Poisson point process models, the idea is to penalize the Poisson likelihood by a penalty function such as $l_1$ penalty \citep[see][]{renner2013equivalence, thurman2014variable}. For more general point process models, instead of employing the likelihood of the processes which often requires computational intensive MCMC methods \citep{moller2003statistical}, penalized versions of estimating equations based on Campbell theorem derived both from Poisson and logistic regression likelihoods have been developed \citep[see][]{thurman2015regularized, choiruddin2017convex}. Furthermore, \cite{thurman2015regularized} and \cite{choiruddin2017convex} show that, under some conditions, the estimates obtained from such procedures are consistent, sparse, and asymptotically normal.

\subsection{Issues in high dimensional data}

The motivation of our paper comes from a study of biodiversity in a 50-hectare region ($D=1,000\mathrm{m} \times  500\mathrm{m}$) of the tropical moist forest of Barro Colorado Island (BCI) in central Panama, where censuses have been carried out such that all free-standing woody stems at least 10 mm diameter at breast height were identified, tagged, and mapped, resulting in maps of over 350,000 individual trees with around 300 species \citep[see][]{condit1998tropical, hubbell1999light, hubbell2005barro}. In the same region, many environmental covariates such as topographical attributes and soil properties have been also collected. In particular, we are interested to study the spatial distribution of 3,604 locations of {\it Beilschmiedia pendula Lauraceae} (BPL) trees and to model its intensity as a parametric function of 93 covariates consisting of 2 topological attributes, 13 soil properties and 78 interactions between two soil nutrients.
	
Although it seems that the number of covariates is not very large with respect to the number of data points, two hours are required to estimate the parameters and select covariates using a standard stepwise procedure. To do this, we use the \texttt{step} function in \texttt{R} which intrinsically assumes that $\mathbf X$ is a Poisson point process since we use the AIC in the stepwise procedure. For a general point process, other criteria could be investigated such as the one based on the $F$ statistic, but they require to estimate the asymptotic covariance matrix of the estimates at each step: even if we know the right covariates, it is known as a difficult task (see e.g. 
\citet{coeurjolly14}) especially when the number of parameters is large. That would easily triple the time of this estimation/selection procedure. To evaluate the performance of such a selection/estimation procedure, a simulation would be required which is unrealistic (1000 replications of a single model would take 250 days). This motivates us to consider regularization methods.
	
\cite{thurman2015regularized} and \cite{choiruddin2017convex} are the first two theoretical works. Both these works have the important limitation that the number of covariates $p$ is finite. We extend this in the present paper. Asymptotic properties which consider a diverging number of parameters for $M$-estimators have a long story \citep[e.g][]{huber1973robust, portnoy1984asymptotic} but have recently been investigated for penalized regression estimators by \citet{fan2004nonconcave,zou2009adaptive}. In particular, as argued by \citet{fan2004nonconcave}, even though the asymptotic properties (i.e., consistency, sparsity, and asymptotic normality) proposed by \cite{fan2001variable} for penalized generalized linear models under the assumption that $p$ is finite, are encouraging, there are many naive and simple model selection procedures which possess those properties. Establishing the validity of these asymptotic properties in a diverging number of parameters setting is, therefore, a major importance. We study this type of asymptotic properties in the spatial point processes framework. Hence, our work can be regarded as an extension of the study conducted by \citet{choiruddin2017convex}.
	
A standard way of measuring asymptotic for spatial point process is the increasing domain asymptotic. Therefore, we investigate the problem where $p=p_n$ grows with $|D_n|$ the volume of the observation domain. In our setting, $|D_n|$ plays the same role as $n$, the number of observations, in standard problems such as in linear models or generalized linear models. We obtain  consistency, sparsity, and asymptotic normality for our estimator. One of our main assumptions is that $p_n^3/|D_n| \to 0$ as $n \to \infty$, which is similar to the one required by \cite{fan2004nonconcave} when $|D_n|$ is simply replaced by $n$ (the sample size in their context).
	
Our results are general: (1) a large choice of penalty functions (either convex or non-convex function) and methods (e.g. ridge, lasso, elastic net, SCAD, and MC+) are available; (2) we include a large class of mixing spatial point processes. The implementation is done by combining the $\texttt{spatstat}$  \citep{baddeley2015spatial} $\texttt{R}$ package with the two $\texttt{R}$ packages implementing penalized methods for generalized linear models: $\texttt{glmnet}$ \citep{friedman2010regularization} and $\texttt{ncvreg}$ \citep{breheny2011coordinate}.

\subsection{Outline of the paper}
In Section~\ref{sec:reg}, we introduce brief background on spatial point processes as well as regularization methods for spatial point processes intensity estimation. Section~\ref{ch2:sec:asymp} presents our asymptotic results. We investigate in Section~\ref{ch2:sec:num} the finite sample performance of our estimates in a simulation study and in an application to tropical forestry datasets. Conclusion and discussion are presented in Section~\ref{conclu}. Proofs of the main results are postponed to Appendices~\ref{ch2:auxLemma}-\ref{ch2:proof2}.

\section {Regularization methods for spatial point processes}
\label{sec:reg}
This section gives brief introduction on spatial point processes and reviews regularization methods for spatial point processes intensity estimation previously studied by \citet{choiruddin2017convex}  when the number of parameters is finite.

Let $\mathbf{X}$ be a spatial point process on $\mathbb{R}^d$. We view $\mathbf{X}$ as a locally finite random subset of $\mathbb{R}^d$. Let $D \subset \mathbb{R}^d$ be a compact set of Lebesgue measure $|D|$ which will play the role of the observation domain. A realization of $\mathbf{X}$ in $D$ is thus a set $\mathbf{x}=\{x_1, x_2, \ldots, x_m\}$, where $x \in D$ and $m$ is the observed number of points in $D$.  Suppose $\mathbf{X}$ has intensity function $\rho$ and second-order product density $\rho^{(2)}$. Campbell theorem \citep[see e.g.][]{moller2003statistical} states that, for any function $k: \mathbb{R}^d \to [0,\infty)$ or $k: \mathbb{R}^d \times \mathbb{R}^d \to [0,\infty)$
\begin{align}
&\mathbb{E} \Big( \sum_{u \in \mathbf{X}} k(u) \Big) ={\int_{\mathbb{R}^d} k(u) \rho (u)\mathrm{d}u} \label{ch2: eq:campbell} \\
&\mathbb{E} \Big(\sum_{u,v \in \mathbf{X}}^{\neq} k(u,v)\Big)=\int_{\mathbb{R}^d}{\int_{\mathbb{R}^d} k(u,v) \rho^{(2)} (u,v)\mathrm{d}u \mathrm{d}v} \label{ch2:eq:campbell2}.
\end{align}
We may interpret $\rho(u) \mathrm{d}u$ as the probability of occurence of a point in an infinitesimally small ball with centre $u$ and volume $\mathrm{d}u$. Intuitively, $\rho^{(2)}(u,v)\mathrm{d}u\mathrm{d}v$ is the probability for observing a pair of distinct points from $\mathbf{X}$ occuring jointly in each of two infinitesimally small balls with centres $u,v$ and volume $\mathrm{d}u, \mathrm{d}v$. For further background materials on spatial point processes, see for example \citet{moller2003statistical, illian2008statistical}.

In our study, we assume that the intensity function depends on parameter $\boldsymbol \beta$, $\rho= \rho(\cdot;\boldsymbol \beta)$. The standard parametric methods for estimating $\boldsymbol \beta$ are by maximizing the weighted Poisson likelihood \citep[e.g.][]{guan2010weighted} or the weighted logistic regression likelihood \citep[e.g.][]{baddeley2014logistic, choiruddin2017convex} given respectively by
\begin{align}
\ell_\mathrm{PL}(w;\boldsymbol \beta)  = & {\sum_{u \in \mathbf{X} \cap D} w(u)\log\rho(u;\boldsymbol \beta)} - {\int_D w(u)\rho(u; \boldsymbol\beta)\mathrm{d}u}, \label{ch2:Pois}    \\
\ell_\mathrm{LRL}(w; \boldsymbol \beta) = & {\sum_{u \in \mathbf{X}  \cap D} w(u) \log \left ( \frac {\rho(u;\boldsymbol \beta)} {\delta(u)+\rho(u;\boldsymbol \beta)} \right )} \nonumber \\
&- {\int_D w(u) \delta(u) \log \left (\frac {\rho(u;\boldsymbol \beta) + \delta(u)} {\delta(u)} \right ) \mathrm{d}u}, \label{ch2:logistic}
\end{align}
where $w(\cdot)$ is a weight non-negative function depending on the first and the second-order characterictics of $\mathbf{X}$ and $\delta(\cdot)$ is a non-negative real-valued function. The solution of maximizing (\ref{ch2:Pois}) (resp. (\ref{ch2:logistic})) is called Poisson estimator (resp. the logistic regression estimator). We refer readers to \cite{guan2010weighted} for further details on the weight function $w(\cdot)$ and to \cite{baddeley2014logistic} for the role of function $\delta(\cdot)$.

These standard methods cannot perform variable selection. To do so, \cite{thurman2015regularized} and \cite{choiruddin2017convex} suggest to maximize a penalized version of (\ref{ch2:Pois})-(\ref{ch2:logistic})
\begin{align}
\label{ch2:regmed}
	Q(w; \boldsymbol \beta)=\ell(w; \boldsymbol \beta) - |D|{\sum_{j=1}^p p_{\lambda_j}(|\beta_{j}|)},
\end{align}
where $\ell(w; \boldsymbol \beta)$ is either the Poisson likelihood (\ref{ch2:Pois}) or the logistic regression likelihood (\ref{ch2:logistic}). We refer the second term of (\ref{ch2:regmed}) to a penalization term. In this term, we have mainly two parts: (1) a penalty function $p_\lambda$ parameterized by $\lambda \geq 0$ and (2) the volume of the observation domain $|D|$ which plays the same role as the sample size in the spatial point process framework.

For any non-negative $\lambda$, we say that $p_\lambda(\cdot)$ is a penalty function if $p_\lambda$ is a non-negative function with $p_\lambda(0)=0$. Some examples, described in Table~\ref{ch2:table:rm}, include  $l_2$ penalty \citep{hoerl1988ridge}, $l_1$ penalty \citep{tibshirani1996regression}, elastic net \citep{zou2005regularization}, SCAD \citep{fan2001variable}, and MC+ \citep{zhang2010nearly}. Note that, as indicated by \eqref{ch2:regmed}, we allow each direction to have different tuning parameters $\lambda_j, j=1,\ldots,p$. Such a method is called an adaptive method (e.g. adaptive lasso \citep{zou2006adaptive} and adaptive elastic net \citep{zou2009adaptive}). For further backgrounds about penalty function and regularization methods, see, for example, \cite{friedman2008elements}.

\setlength{\tabcolsep}{3pt}
\renewcommand{\arraystretch}{1.5}
\begin{table}[htb]
\caption{Examples of penalty function.}
\label{ch2:table:rm}
\centering
\begin{threeparttable}
\begin{tabular}{ll}
\hline
\hline
Penalty & $ p_{\lambda}(\theta)$ \\
\hline
\hline
$l_2$ penalty &  $ \frac{1}{2} \lambda \theta^2$ \\
$l_1$ penalty & $\lambda |\theta| $ \\
Enet & $ \lambda \{\gamma |\theta|+\frac{1}{2} (1-\gamma)\theta^2\}$, for any $0 < \gamma<1 $ \\ 
\multirow{2}{*}{SCAD} & $ \lambda \theta \mathbb{I}(\theta \leq \lambda)+ \frac{\gamma \lambda \theta - \frac {1}{2} (\theta^2 + \lambda^2)}{\gamma-1} \mathbb{I}(\lambda \leq \theta \leq \gamma \lambda) + \frac{\lambda^2(\gamma^2-1)}{2(\gamma-1)}\mathbb{I}(\theta \geq \gamma \lambda)$, \\
& for any $\gamma>2 $\\
MC+ &  $ \Big(\lambda \theta - \frac{\theta^2}{2 \gamma}\Big) \mathbb{I}(\theta \leq  \gamma \lambda)+\frac{1}{2}\gamma \lambda^2 \mathbb{I}(\theta \geq \gamma \lambda)$, for any $\gamma>1 $\\
\hline
\end{tabular}
\end{threeparttable}
\end{table}

\section{Asymptotic properties}
\label{ch2:sec:asymp}
In this section, we present asymptotic properties of the regularized Poisson estimator when both $|D_n| \to \infty$ and $p_n \to \infty$ as $n \to \infty$. In particular, we consider $\mathbf{X}$ as a $d$-dimensional point process observed over a sequence of observation domain $D=D_n, n=1,2, \ldots$ which expands to $\mathbb{R}^d$ as $n \to \infty$. We assume that $\mathbf{X}$ has a log-linear form  given by (\ref{ch2:intensity function}) for which the dimension of parameter $\boldsymbol \beta$, denoted now by $p_n$, diverges to $\infty$ as $n \to \infty$.  In Section \ref{ch2:sec:not}, we provide notation and conditions, and discuss the differences   from the setting where $p$ is fixed. Our main results are presented in Section \ref{sec:result}.  For sake of conciseness, we do not present the asymptotic results for the regularized logistic regression estimator. The results are very similar. The main difference is lying in the conditions ($\mathcal C$.\ref{ch2:C:BnCn}) and ($\mathcal C$.\ref{ch2:C:An}) for which the matrices  $\mathbf{A}_n, \mathbf{B}_n,$ and $\mathbf{C}_n$ have a different expression (see Remark \ref{ch2:Rlogistic}).

\subsection{Notation and conditions} \label{ch2:sec:not}
Throughout this section and Appendices~\ref{ch2:auxLemma}-\ref{ch2:proof2}, let 
\begin{align}
\ell_n (w;\boldsymbol \beta)=& \ell_{n,\mathrm{PL}}(w;\boldsymbol \beta) \nonumber \\
 =& {\sum_{u \in \mathbf{X} \cap D_n} w(u)\log\rho(u;\boldsymbol \beta)} - {\int_{D_n} w(u)\rho(u; \boldsymbol\beta)\mathrm{d}u},  \label{ln}\\
Q_n(w;\boldsymbol \beta)= & \ell_n (w;\boldsymbol \beta) -  |D_n| {\sum_{j=1}^{p_n} p_{\lambda_{n,j}}(|\beta_j|)} \label{Qn},
\end{align}
be respectively the weighted Poisson likelihood and its penalized version.

Let $\boldsymbol \beta_0 = \{\beta_{01},\ldots,\beta_{0s},\beta_{0(s+1)},\ldots,\beta_{0p_n}\}^\top=\{\boldsymbol \beta^{\top}_{01},\boldsymbol \beta^{\top}_{02}\}^\top = (\boldsymbol \beta_{01}^\top, \mathbf 0^\top)^\top$ denote the $p_n$-dimensional vector to estimate, where ${\boldsymbol \beta_{01}}$ is the $s$-dimensional vector of non-zero coefficients and $\boldsymbol \beta_{02}$ is the ($p_n-s$)-dimensional vector of zero coefficients. We assume that the number of non-zero coefficients, $s$, does not depend on $n$.  Let $\mathbf{z}_{01}$ and $\mathbf{z}_{02}$ denote the corresponding $s$-dimensional and ($p_n - s$)-dimensional vectors of spatial covariates. We denote the regularized Poisson estimator by ${\boldsymbol {\hat { \beta}}}=(\boldsymbol{\hat \beta}_1^\top, \boldsymbol{\hat \beta}_2^\top)^\top $. 


We recall the classical definition of strong mixing coefficients adapted to spatial point processes \citep[e.g.][]{politis1998large}: for $k,l \in \mathbb{N} \cup \{ \infty \}$ and $q \geq 1$, define
\begin{align}
\alpha_{k,l}(q)=\sup\{
&|\mathrm{P}(A \cap B)-\mathrm{P}(A)\mathrm{P}(B)|: A\in \mathscr{F} (\Lambda_1), B \in \mathscr{F} (\Lambda_2), \nonumber\\
&\Lambda_1 \in \mathscr{B} (\mathbb{R}^d), \Lambda_2 \in \mathscr{B} (\mathbb{R}^d), |\Lambda_1| \leq k, |\Lambda_2| \leq l, d(\Lambda_1, \Lambda_2) \geq q \}, \label{ch2:eq:5}
\end{align}
where $\mathscr{F}$ is the $\sigma$-algebra generated by $\mathbf{X} \cap \Lambda_i, i=1,2, d(\Lambda_1, \Lambda_2)$ is the minimal distance between sets $\Lambda_1$ and $\Lambda_2$, and $\mathscr{B}(\mathbb{R}^d)$ denotes the class of Borel sets in $\mathbb{R}^d$.

We  define the $p_n \times p_n$ matrices $\mathbf{A}_n(w;\boldsymbol{\beta}_{0}), \mathbf{B}_n(w;\boldsymbol{\beta}_{0})$ and $\mathbf{C}_n(w;\boldsymbol{\beta}_{0})$ by
\begin{align*}
\mathbf{A}_n(w;\boldsymbol{\beta}_{0})&={\int_{D_n} w(u)\mathbf{z}(u)\mathbf{z}(u)^\top \rho(u;\boldsymbol{\beta}_{0})\mathrm{d}u}, \\
\mathbf{B}_n(w;\boldsymbol{\beta}_{0})&={\int_{D_n} w(u)^2 \mathbf{z}(u)\mathbf{z}(u)^\top \rho(u;\boldsymbol{\beta}_{0})\mathrm{d}u}, \\
\mathbf{C}_n(w;\boldsymbol{\beta}_{0})&={\int_{D_n} \int_{D_n} w(u) w(v)\mathbf{z}(u)\mathbf{z}(v)^\top \{g(u,v)-1\} \rho(u;\boldsymbol{\beta}_{0}) \rho(v;\boldsymbol{\beta}_{0}) \mathrm{d}u \mathrm{d}v},
\end{align*}
where $g(u,v)$ is the classical pair correlation function \cite[][]{moller2003statistical} given by
\begin{align*}
g(u,v)=\frac{\rho^{(2)}(u,v)}{\rho(u)\rho(v)},
\end{align*}
when both $\rho$ and $\rho^{(2)}$ exist with the convention $0/0=0$. For a Poisson point process, we have $g(u,v)=1$ since $\rho^{(2)}(u,v)=\rho(u)\rho(v)$. If, for example, $g(u,v)>1$ (resp. $g(u,v)<1$), this indicates that pair of points are more likely (resp. less likely) to occur at locations $u,v$ than for a Poisson point process. 

We denote the $s \times s$ top-left corner of $\mathbf{A}_{n}(w;\boldsymbol{\beta}_{0})$  $(\mbox{resp. } \mathbf{B}_{n}(w;\boldsymbol{\beta}_{0})$, $ \mathbf{C}_{n}(w;\boldsymbol{\beta}_{0}))$ by $\mathbf{A}_{n,11}(w;\boldsymbol{\beta}_{0})$ $(\mbox{resp. } \mathbf{B}_{n,11} (w;\boldsymbol{\beta}_{0}), \mathbf{C}_{n,11}(w;\boldsymbol{\beta}_{0}))$. It is worth noticing that $\mathbf{A}_{n,11}(w;\boldsymbol{\beta}_{0})$, $\mathbf{B}_{n,11}(w;\boldsymbol{\beta}_{0})$ and $\mathbf{C}_{n,11}(w;\boldsymbol{\beta}_{0})$ depend on $n$ only through $D_n$ and not through $p_n$. In what follows, for a squared symmetric matrix $\mathbf{M}_n$, $\nu_{\min}(\mathbf M_n)$ and $\nu_{\max}(\mathbf M_n)$ denote respectively the smallest and largest eigenvalue of $\mathbf M_n$.

Under the conditions ($\mathcal C$.\ref{ch2:C:plambda})-($\mathcal C$.\ref{ch2:C9}), we define the sequences $a_n$, $b_n$ and $c_n$ by
\begin{align}	
a_n &=\max_{j=1,\ldots,s} | p'_{\lambda_{n,j}}(|\beta_{0j}|)| , \label{ch2:eq:an} \\
b_n &=\inf_{j=s+1,\ldots,p_n} \inf_{\substack{|\theta| \leq \epsilon_n \\ \theta \neq 0}} p'_{\lambda_{n,j}}(\theta) \label{ch2:eq:bn}, \mbox{ for } \epsilon_n=K_1 \sqrt{\frac{p_n}{|D_n|}},	\\
c_n &=  \max_{j=1,\ldots,s} |p^{\prime\prime}_{\lambda_{n,j}}(|\beta_{0j}|) |,  \label{ch2:eq:cn}
\end{align}
where $K_1$ is any positive constant. 

Consider the following conditions ($\mathcal C$.\ref{ch2:C:Dn})-($\mathcal C$.\ref{ch2:C9}) which are required to derive our asymptotic results:

\begin{enumerate}[($\mathcal C$.1)]
\item  For every $n \geq 1, D_n=nE=\{ne: e \in E\}$, where $E \subset \mathbb{R}^d$ is convex, compact, and contains the origin of $\mathbb{R}^d$ in its interior. \label{ch2:C:Dn}
\item The intensity function has the log-linear specification given by~\eqref{ch2:intensity function} where $\boldsymbol \beta \in \Theta$ and $\Theta$ is an open convex bounded set of $\mathbb R^{p_n}$. Furthermore, we assume that there exists a neighborhood $\Xi(\boldsymbol \beta_0)$ of $\boldsymbol \beta_0$ such that 
\[
\sup_{n\geq 1} \; \sup_{\boldsymbol \beta \in \Xi(\boldsymbol \beta_0)} \; \sup_{u \in \mathbb{R}^d} \rho(u; \boldsymbol \beta) < \infty.
\]  \label{ch2:C:Theta}
\item  The covariates $\mathbf z$ and the weight function $w$ satisfy
\[
 	\sup_{n\geq 1} \; \sup_{i=1,\dots,p_n} \; \sup_{u \in \mathbb{R}^d} |z_i(u)| < \infty, \; 
 	\quad \mbox{ and } \quad 
 	\sup_{u \in \mathbb{R}^d} w(u)<\infty.
 \] \label{ch2:C:cov}
\item  There exists an integer $t \geq 1$ such that for $k=2, \ldots, 2+t$, the product density $\rho^{(k)}$ exists and satisfies $\rho^{(k)}< \infty$. \label{ch2:C:rhok}
\item  For the strong mixing coefficients (\ref{ch2:eq:5}), we assume that there exists some $\tilde t > d(2+t)/ t$ such that $\alpha_{2,\infty}(q)=O(q^{-\tilde t})$. \label{ch2:C:mixing}
\item $\liminf_{n} \nu_{\min}\big(|D_n|^{-1}\{\mathbf{B}_{n,11}(w;\boldsymbol \beta_0)+\mathbf{C}_{n,11}(w;\boldsymbol \beta_0)\} \big)>0$. \label{ch2:C:BnCn}
\item $\liminf_{n}\nu_{\min}\big(|D_n|^{-1}\mathbf{A}_{n}(w;\boldsymbol \beta_0)\big)> 0$.  \label{ch2:C:An}
\item   The penalty function $p_{\lambda}(\cdot)$ is non-negative on $\mathbb R^+$, continuously differentiable on $\mathbb R^+ \setminus\{0\}$ with derivative ${p}_\lambda'$ assumed to be a Lipschitz function on $\mathbb R^+\setminus\{0\}$.
Furthermore, given $(\lambda_{n,j})_{n \geq 1}, \mbox{ for } j=1, \ldots, s,$ we assume that there exists $(\tilde r_{n,j})_{n \geq 1}$, where $\tilde r_{n,j}\sqrt{|D_n|/p_n} \to \infty$ as $n \to \infty$, such that, for $n$ sufficiently large, $p_{\lambda_{n,j}}$ is thrice continuously differentiable in the ball centered at $|\beta_{0j}|$ with radius $\tilde r_{n,j}$ and we assume that the third derivative is uniformly bounded. \label{ch2:C:plambda}
\item $p_n^3 / |D_n| \to 0$ as $n \to \infty$. \label{ch2:C9}
\end{enumerate}

Conditions ($\mathcal C$.\ref{ch2:C:Dn})-($\mathcal C$.\ref{ch2:C:plambda}) are quite similar to the ones required by \citet{choiruddin2017convex} in the setting when the number of parameters to estimate is fixed. Condition~($\mathcal C$.\ref{ch2:C:Theta}) is slightly stronger since we have to ensure that $ \rho(u; \boldsymbol \beta)$ is finite for $\boldsymbol \beta$ in the neighborhood of $\boldsymbol \beta_0$. Note that $\sup_{u \in \mathbb{R}^d} \rho(u; \boldsymbol \beta_0) < \infty$ follows directly from condition~($\mathcal C$.\ref{ch2:C:cov}). We derive asymptotic properties when both $|D_n|$ and $p_n$ tend to infinity with $n$. However, to obtain an estimator which is consistent and has two other properties: sparsity and asymptotic normality, we need that the number of covariates does not grow too fast with respect to the volume of the observation domain. This condition is  stated by condition~($\mathcal C$.\ref{ch2:C9}) which is similar to the one required by \cite{fan2004nonconcave} when $|D_n|$ is simply replaced by $n$ (the sample size in their context).

\subsection{Main results} \label{sec:result}
We state our main results here. Proofs are relegated to Appendices~\ref{ch2:auxLemma}-\ref{ch2:proof2}.

We first show in Theorem~\ref{ch2:the1} that the regularized Poisson estimator converges in probability and exhibits its rate of convergence.

\begin{theorem}
\label{ch2:the1}
Assume the conditions ($\mathcal C$.\ref{ch2:C:Dn})-($\mathcal C$.\ref{ch2:C:mixing}) and ($\mathcal C$.\ref{ch2:C:An})-($\mathcal C$.\ref{ch2:C9}) hold. Let $a_n$ and $c_n$ be given respectively by (\ref{ch2:eq:an}) and~\eqref{ch2:eq:cn}. If $a_n=O(|D_n|^{-1/2})$ and $c_n=o(1)$, then there exists a local maximizer ${\boldsymbol {\hat \beta}}$ of $Q_n(w;\boldsymbol \beta)$  such that ${\bf \| \boldsymbol {\hat \beta} -\boldsymbol \beta_0\|}=O_\mathrm{P}\big(\sqrt {p_n} (|D_n|^{-1/2}+a_n)\big)$.
\end{theorem}

This implies that, the regularized Poisson estimator is root-$(|D_n|/p_n)$ consistent. Note that, as expected, the convergence rate is $\sqrt{p_n}$ times the convergence rate of the estimator obtained when $p$ is fixed \citep[see Theorem 1][]{choiruddin2017convex}. In addition, when we compare our results with the ones obtained by \cite{fan2004nonconcave}, who also considered a diverging number of parameters setting, our estimator has the same rate of convergence when we replace $|D_n|$ by $n$ to their context. This rate of convergence also appears in other contexts considering diverging number of parameters setting \cite[see e.g.][]{lam2008profile, zou2009adaptive, li2011nonconcave, cho2013model, wang2017variable}.


Now, we demonstrate in Theorem~\ref{ch2:the2} that such a root-$(|D_n|/p_n)$ consistent estimator ensures the sparsity of $\boldsymbol{\hat \beta}$; that is, the estimate will correctly set $\boldsymbol \beta_2$ to zero with probability tending to 1 as $n \to \infty$, and $\boldsymbol{\hat \beta}_1$ is asymptotically normal.

\begin{theorem}
\label{ch2:the2}
Assume the conditions ($\mathcal C$.\ref{ch2:C:Dn})-($\mathcal C$.\ref{ch2:C9}) are satisfied. If $a_n \sqrt{|D_n|}\to \nolinebreak 0$, \linebreak $b_n \sqrt{|D_n|/p_n^2} \to \infty$ and $ c_n \sqrt{p_n} \to 0$ as $n\to\infty$, the root-$(|D_n|/p_n)$ consistent local maximizer ${ \boldsymbol {\hat { \beta}}}=(\boldsymbol{\hat \beta}_1^\top, \boldsymbol{\hat \beta}_2^\top)^\top $ in Theorem~\ref{ch2:the1} satisfies:
\begin{enumerate}[(i)]
\item Sparsity: $\mathrm{P}(\boldsymbol{\hat \beta}_2=0) \to 1$ as $n \to \infty$,
\item Asymptotic Normality: $|D_n|^{1/2} \boldsymbol \Sigma_n(w;\boldsymbol{\beta}_{0})^{-1/2}(\boldsymbol{\hat \beta}_1- \boldsymbol{\beta}_{01})\xrightarrow{d} \mathcal{N}(0, \mathbf{I}_{s})$,
\end{enumerate}
where
\begin{align}
\boldsymbol \Sigma_n(w;\boldsymbol{\beta}_{0})= & |D_n|\{\mathbf{A}_{n,11}(w;\boldsymbol{\beta}_{0})+|D_n| \boldsymbol \Pi_n \}^{-1}\{\mathbf{B}_{n,11}(w;\boldsymbol{\beta}_{0})+\mathbf{C}_{n,11}(w;\boldsymbol{\beta}_{0})\}\nonumber \\
& \{\mathbf{A}_{n,11}(w;\boldsymbol{\beta}_{0})+|D_n| \boldsymbol \Pi_n \}^{-1}, \label{ch2:eq:Sigman} \\
\boldsymbol \Pi_n = & \mathrm{diag}\{p''_{\lambda_{n,1}}(|\beta_{01}|),\ldots,p''_{\lambda_{n,s}}(|\beta_{0s}|)\}. \label{ch2:eq:pi}
\end{align}
As a consequence, $\boldsymbol \Sigma_n(w;\boldsymbol{\beta}_{0})$ is the asymptotic covariance matrix of $\boldsymbol{\hat \beta}_1$. Here, $\boldsymbol \Sigma_n(w;\boldsymbol{\beta}_{0})^{-1/2}$ is the inverse of $\boldsymbol \Sigma_n(w;\boldsymbol{\beta}_{0})^{1/2}$, where $\boldsymbol \Sigma_n(w;\boldsymbol{\beta}_{0})^{1/2}$ is any square matrix with $\boldsymbol \Sigma_n(w;\boldsymbol{\beta}_{0})^{1/2}\big(\boldsymbol \Sigma_n(w;\boldsymbol{\beta}_{0})^{1/2}\big)^\top=\boldsymbol \Sigma_n(w;\boldsymbol{\beta}_{0})$.
\end{theorem}

\begin{remark}
\label{ch2:remark 1}
For lasso and adaptive lasso, $ \boldsymbol \Pi_n= \mathbf 0$. For other penalties, since $c_n=o(1)$, then $\boldsymbol \|  \boldsymbol \Pi_n\|=o(1)$. Since $\|\mathbf{A}_{n,11}(w;\boldsymbol{\beta}_{0})\|=O(|D_n|)$ from conditions ($\mathcal C$.\ref{ch2:C:Dn})-($\mathcal C$.\ref{ch2:C:cov}), $|D_n| \, \|  \boldsymbol \Pi_n\|$ is asymptotically negligible with respect to $\|\mathbf{A}_{n,11}(w;\boldsymbol{\beta}_{0})\|$.
\end{remark}

\begin{remark}
\label{ch2:Rlogistic}
Theorems \ref{ch2:the1} and \ref{ch2:the2} remain true for the regularized logistic regression estimator if we replace in the expression of the matrices $\mathbf{A}_n, \mathbf{B}_n,$ and $\mathbf{C}_n$, $w(u)$ by ${w(u) \delta(u)}/ ({\rho(u; \boldsymbol \beta_0) + \delta(u)}), u \in D_n$ and extend the condition ($\mathcal C$.\ref{ch2:C:cov}) by adding $\sup_{u \in \mathbb{R}^d} \delta(u)<\infty$.
	
The proofs of Theorems \ref{ch2:the1} and \ref{ch2:the2} for this estimator are slightly different mainly because unlike the Poisson likelihood for which we have $\ell_n^{2}(w;\boldsymbol{\beta})=-\mathbf{A}_{n}(w;\boldsymbol{\beta})$, for the regularized logistic regre $\ell_n^{2}(w;\boldsymbol{\beta})$ is now stochastic and we only have $\mathbb{E}(\ell_n^{2}(w;\boldsymbol{\beta}))=-\mathbf{A}_{n}(w;\boldsymbol{\beta})$. Despite the additional difficulty, we maintain that no additional assumption is required. 
\end{remark}

We show in Theorem \ref{ch2:the2} that the sparsity and asymptotic normality are still valid when the number of parameters diverges. By Remark~\ref{ch2:remark 1}, when $n$ is large enough, $\boldsymbol \Sigma_n(w;\boldsymbol{\beta}_{0})$ in \eqref{ch2:eq:Sigman} becomes approximately
\begin{align*}
|D_n|\{\mathbf{A}_{n,11}(w;\boldsymbol{\beta}_{0})\}^{-1}\{\mathbf{B}_{n,11}(w;\boldsymbol{\beta}_{0})+\mathbf{C}_{n,11}(w;\boldsymbol{\beta}_{0})\} \{\mathbf{A}_{n,11}(w;\boldsymbol{\beta}_{0}) \}^{-1},
\end{align*}
which is precisely the asymptotic covariance matrix of the estimator of $\boldsymbol \beta_{01}$ obtained by maximizing the likelihood function or solving estimating equations based on the submodel knowing that $\boldsymbol \beta_{02}=\mathbf{0}$. This shows that when $n$ is sufficiently large, our estimator is as efficient as the oracle one.

To satisfy Theorem~\ref{ch2:the2}, we require that $a_n \sqrt{|D_n|} \to 0$, $b_n \sqrt{|D_n|/p_n^2} \to \infty$ and $ c_n \sqrt{p_n}\to 0$ as $n \to \infty$ simultaneously. In particular, conditions on $a_n$ and $c_n$ ensure the asymptotic normality of $\boldsymbol{\hat \beta}_1$ while condition on $b_n$ is used to prove the sparsity. Conditions regarding $a_n$ and $c_n$ are similar to the ones imposed by \cite{fan2004nonconcave} when $|D_n|$ is replaced by $n$ in their context. However, we require a slightly stronger condition on $b_n$ than the one required by \cite{fan2004nonconcave} which in the present setting could be written as $b_n \sqrt{|D_n|/p_n} \to \infty$. As compensation, we do not need to impose, as \cite{fan2004nonconcave} did, for any $0 < K_2 < \infty$,  $\nu_{\max}\big(|D_n|^{-1}\mathbf{A}_{n}(w;\boldsymbol \beta_0)\big) < K_2$. Such a condition is not straightforwardly satisfied in our setting since the other conditions only imply that $\nu_{\max}\big(|D_n|^{-1}\mathbf{A}_{n}(w;\boldsymbol \beta_0)\big)=O({p_n})$.

Further details regarding $a_n$, $b_n$ and $c_n$ for each method are presented in Table~\ref{ch2:anbncn}. For the ridge regularization method, $b_n=0$, preventing from applying Theorem~\ref{ch2:the2} for this penalty. For lasso and elastic net, $a_n=K_3 b_n$ for some constant $K_3>0$ ($K_3$=1 for lasso). The two conditions $a_n \sqrt{|D_n|} \to 0$ and $b_n \sqrt{|D_n|/p_n^2} \to \infty$ as $n \to \infty$ cannot be satisfied simultaneously. This is different for the adaptive versions where a compromise can be found by adjusting the $\lambda_{n,j}$'s, as well as the two non-convex penalties SCAD and MC+, for which $\lambda_n$ can be adjusted. For the regularization methods we consider in this study, the condition $ c_n \sqrt{p_n}\to 0$ is implied by the condition $a_n \sqrt{|D_n|}\to 0$ as $n\to \infty$ and condition~($\mathcal C$.\ref{ch2:C9}).

\setlength{\tabcolsep}{3pt}
\renewcommand{\arraystretch}{1.5}
\begin{table}[!htb]
\caption{Details of the sequences $a_n$, $b_n$ and $c_n$ for a given regularization method.}
\label{ch2:anbncn}
\begin{center}
\begin{threeparttable}
\begin{tabular}{ l l l l }
\hline
\hline
Method &$a_n$ & $b_n$ & $c_n$\\
\hline
\hline
Ridge & $\lambda_n {\displaystyle \max_{j=1,...s} \{|\beta_{0j}|\}}$  & $0$ & $\lambda_n$ \\
Lasso & $\lambda_n$ & $\lambda_n$ & 0  \\
Enet & $\lambda_n \left[(1-\gamma){\displaystyle \max_{j=1,...s} \{|\beta_{0j}|\}}+\gamma \right]$  &$\gamma \lambda_n$& $(1-\gamma) \lambda_n$\\
AL & ${\displaystyle \max_{j=1,...s} \{\lambda_{n,j}\}}$ & ${\displaystyle \min_{j=s+1,...p_n} \{\lambda_{n,j}\}}$ & 0\\ 
Aenet & $ {\displaystyle \max_{j=1,...s}\{ \lambda_{n,j}\big((1-\gamma) |\beta_{0j}|+\gamma\big) \}}$  & $\gamma {\displaystyle \min_{j=s+1,...p_n}\{\lambda_{n,j} \}}$ & $(1-\gamma){\displaystyle\max_{j=1,\dots,s} \{\lambda_{n,j}\}}$\\ 
SCAD & $0 \tnote{*}$ & $\lambda_n  \tnote{**}$  & $0\tnote{*}$\\ 
MC+ & $0  \tnote{*}$ & $\lambda_n - \frac {K_1 \sqrt{p_n}} {\gamma \sqrt{|D_n|}}  \tnote{**}$ &$0\tnote{*}$\\ 
\hline
\end{tabular}
\begin{tablenotes}
\item[*] if $\lambda_n \to 0$ for $n$ sufficient large 
\item[**] if $ \lambda_n \sqrt{|D_n|/p_n^2} \to \infty$ for $n$ sufficient large
\end{tablenotes}
\end{threeparttable}
\end{center}
\end{table}

\section{Numerical results}
\label{ch2:sec:num}
This section is devoted to present numerical results. More precisely, we conduct simulation experiments in Section~\ref{ch2:simul} to assess the finite sample peformance of our estimates and apply our method to an application in ecology in Section~\ref{ch2:sec8}. We apply the regularized Poisson likelihood (PL) and the regularized weighted Poisson likelihood (WPL) to select covariates and estimate their coefficients. Similar approach can be straightforwardly used for the regularized versions using logistic regression likelihood.

To numerically evaluate the parameters estimates, we apply Berman-Turner method \citep{berman1992approximating} combined with coordinate descent algorithm  \citep{friedman2007pathwise} to perform variable selection and parameter estimation. Berman-Turner device allows to show that maximizing \eqref{ch2:Pois} is equivalent to fitting a weighted Poisson generalized linear model, so the standard software for generalized linear models (GLMs) can be used. This  has been exploited by the $\texttt{spatstat}$  $\texttt{R}$ package \citep{baddeley2015spatial}. As we make links between spatial point processes intensity estimation and GLMs, we only have to deal with feature selection procedures for GLMs. Hence, we clearly have many advantages: the various computational strategies are carefully studied, and, in particular, efficiently implemented in $\texttt{R}$. In this study, to compute the regularization path solutions, we employ coordinate descent algorithm \citep{friedman2007pathwise}. This is implemented in the $\texttt{glmnet}$ \citep{friedman2010regularization} for regularization methods for GLMs using some convex penalties (i.e., ridge, lasso, elastic net, adaptive lasso and adaptive elastic net) and in the $\texttt{ncvreg}$ \citep{breheny2011coordinate} for regularization methods for GLMs using some non-convex penalties (i.e., SCAD and MC+). More details for computational strategies are discussed in detail  by \citet{choiruddin2017convex}.

Our methods rely on the tuning parameter $\lambda$. Some previous studies  \citep[see e.g.][]{zou2007degrees, wang2007tuning, wang2009shrinkage} suggest to use a modified BIC criterion to select the tuning parameter. We follow the literature and choose $\lambda$ by minimizing $\mathrm{WQBIC}(\lambda)$, a modified version of the BIC criterion, defined by
\begin{align*}
\mathrm{WQBIC}(\lambda)=-2 \ell(w;\boldsymbol{\hat \beta} (\lambda))+ s(\lambda)\log|D|,
\end{align*}
where $s(\lambda)={\sum_{j=1}^p \mathbb{I}\{{\hat \beta_j}(\lambda) \neq 0\}}$ is the number of selected covariates with non-zero regression coefficients and $|D|$ is the volume of observation domain. To implement the adaptive methods (i.e., adaptive lasso and adaptive elastic net), we follow \cite{zou2006adaptive} and define $\lambda_j=\lambda / | \tilde{\beta_j}(ridge)|, j =1, \cdots,p$, where $ \tilde{\boldsymbol \beta}(ridge)$ is the estimates obtained from ridge regression and $\lambda$ is a tuning parameter chosen by $\mathrm{WQBIC}(\lambda)$ criterion as described above. Following \citet{choiruddin2017convex}, we fix $\gamma=0.5$ for elastic net and its adaptive version, $\gamma=3.7$ for SCAD, and $\gamma=3$ for MC+. For further discussion regarding the selection of $\gamma$ for SCAD and MC+, see e.g. \cite{fan2001variable} and \cite{breheny2011coordinate}.

\subsection{Simulation study}
\label{ch2:simul}
In this section, we investigate the behavior of our estimators in a simulation experiment in different situations when a large number of covariates for fitting spatial point process intensity estimation is involved. We intend to extend the setting considered by \citet{choiruddin2017convex}. We start with relatively complex situation where strong multicollinearity is present (Scenarios~\ref{sce1a} and \ref{sce2a}) and we then consider a more complex setting using real datasets (Scenarios~\ref{sce1b} and \ref{sce2b}). We have two different scenarios (Scenarios~\ref{ch2:sce1} and \ref{ch2:sce2}) for which the number of true covariates as well as their coefficients are different.

The spatial domain we consider is $D=[0,1000] \times[0,500]$. The true intensity function has the form $\rho(u; \boldsymbol \beta_0)=\exp(\mathbf{z}(u)^\top \boldsymbol \beta_0)$, where $\mathbf{z}(u)=\{1, z_1(u),\ldots,z_{50}(u)\}^\top$ and $\boldsymbol \beta_0=\{\beta_0, \beta_{01}, \cdots, \beta_{050}\}$. We set  $\beta_0$ such that the mean number of points over $D$ is equal to $1600$. We consider two different scenarios described as follows.
\begin{enumerate}[{Scenario}~1.]
\item We define the true vector $\boldsymbol \beta_0=\{\beta_0, 2, 0.75, 0, \cdots,0\}$. To define the covariates, we center and scale the $201 \times 101$ pixel images of elevation ($x_1$) and gradient of elevation ($x_2$) contained in the $\texttt{bei}$ datasets of $\texttt{spatstat}$ library in $\texttt{R}$ 
 and use them as two true covariates. In addition, we create two settings to define extra covariates: \label{ch2:sce1}
\begin{enumerate}[a.]
\item  First, we generate 48 $201 \times 101$ pixel images of covariates as a standard Gaussian white noise and denote them by $x_3, \ldots, x_{50}$. Second, we transform them, together with $x_1$ and $x_2$, to have multicollinearity. In particular, we define $\mathbf{\tilde z}(u)=\mathbf{V} ^\top \mathbf{x}(u)$, where $\mathbf{x}(u)=\{x_1(u),\ldots,x_{50}(u)\}^\top$. More precisely, $\mathbf{V}$ is such that $\boldsymbol \Omega=\mathbf{V} ^\top \mathbf{V},$ and $(\Omega)_{ij}=(\Omega)_{ji}=0.7^{|i-j|}$ for $i,j=1,\ldots,50$, except $(\Omega)_{12}=(\Omega)_{21}=0$, to preserve the correlation between $x_1$ and $x_2$.   \label{sce1a} 
In this setting, $\mathbf{z}(u)=\{1,\mathbf{\tilde z}(u)\}$.

\item  We center and scale the 13 $50 \times 25$ pixel images of soil nutrients obtained from the study in tropical forest of Barro Colorado Island (BCI) in central Panama \citep[see][]{condit1998tropical, hubbell1999light, hubbell2005barro} and convert them to be  $201 \times 101$ pixel images as $x_1$ and $x_2$. In addition, we consider the interaction between two soil nutrients such that we have 50 covariates in total. We use 48 covariates (13 soil nutrients and 35 interactions between them) as the extra covariates. Together with $x_1$ and $x_2$, we keep the structure of the covariance matrix to preserve the complexity of the situation. In this setting, we have $\mathbf{z}(u)=\mathbf{x}(u)=\{1, x_1(u),\ldots,x_{50}(u)\}^\top$.   \label{sce1b} 
\end{enumerate}

\item In this setting, we consider five true covariates out of 50 covariates. In addition of elevation ($x_1$) and gradient of elevation ($x_2$), we convert $50 \times 25$ pixel images of concentration of Aluminium ($x_3$), Boron ($x_4$) and Calcium ($x_5$) in the soil to be $201 \times 101$ pixel images as $x_1$ and $x_2$ and set them to be other three true covariates. All five covariates are centered and scaled. We define the true coefficient vector $\boldsymbol \beta_0=\{\beta_0, 5, 4,  3, 2, 1, 0, \cdots,0\}$. As in Scenario~\ref{ch2:sce1}, we make two settings to define 45 extra covariates: \label{ch2:sce2}
\begin{enumerate}[a.]
\item This setting is similar to that of Scenario~\ref{sce1a}. We generate 45 $201 \times 101$ pixel images of covariates as standard Gaussian white noise, denote them by $x_6, \ldots, x_{50}$, and define $\mathbf{\tilde z}(u)=\mathbf{V} ^\top \mathbf{x}(u)$, where $\mathbf{V}$ is such that $\boldsymbol \Omega=\mathbf{V} ^\top \mathbf{V},$ and $(\Omega)_{ij}=(\Omega)_{ji}=0.7^{|i-j|}$ for $i,j=1,\ldots,50$, except $(\Omega)_{kl}=(\Omega)_{lk}=0,$ for $k,l=1,\cdots,5, k \neq l$, to preserve the correlation among $x_1$ - $x_5$. We still define $\mathbf{z}(u)=\{1,\mathbf{\tilde z}(u)\}$. \label{sce2a}

\item  We use the real dataset as in Scenario~\ref{sce1b} and consider similar setting. In this setting, we define 5 true covariates which have different regression coefficients as in  Scenario~\ref{sce1b}. \label{sce2b}
\end{enumerate}
\end{enumerate}

\setlength{\tabcolsep}{2.5pt}
\renewcommand{\arraystretch}{1.5}
\begin{table}[!ht]
\caption{Empirical selection properties (TPR, FPR, and PPV in $\%$) based on 2000 replications of Thomas processes on the domain $D$ for two different values of $\kappa$ and for the two different scenarios. Different penalty functions are considered as well as two estimating equations, the regularized Poisson likelihood (PL) and the regularized weighted Poisson likelihood (WPL).}
\label{ch2:table:sel}
\centering
\begin{threeparttable}
\scalebox{0.77}{
\begin{tabular}{ l | ccc | ccc | ccc | ccc}
\hline
\multirow{3}{*}{Method}  & \multicolumn{3}{c}{Regularized PL} & \multicolumn{3}{c}{Regularized WPL} & \multicolumn{3}{c}{Regularized PL} & \multicolumn{3}{c}{Regularized WPL} \\
\cline{2-13}
 & \multicolumn{6}{c}{$\kappa=5 \times 10^{-4}$} & \multicolumn{6}{c}{$\kappa=5 \times 10^{-5}$} \\
\cline{2-13}
 & TPR & FPR & PPV & TPR & FPR & PPV & TPR & FPR & PPV & TPR & FPR & PPV \\ 
\hline
 & \multicolumn{12}{c}{Scenario~\ref{sce1a}} \\
\hline
  Lasso & 100\tnote{1} & 13 & 28 & 96 & 4 & 62 & 97 & 23 & 20 & 64 & 1 & 76 \\ 
  Enet & 100\tnote{1} & 34 & 12 & 93 & 8 & 48 & 97 & 48 & 10 & 59 & 2 & 58 \\ 
  AL& 100\tnote{1} & 1 & 92 & 97 & 0\tnote{1} & 96 & 95 & 3 & 68 & 70 & 0\tnote{1} & 98 \\ 
  Aenet & 100\tnote{1} & 2 & 76 & 97 & 1 & 85 & 95 & 6 & 52 & 67 & 0\tnote{1} & 95 \\ 
  SCAD & 100\tnote{1} & 7 & 41 & 97 & 1 & 87 & 96 & 4 & 61 & 56 & 0\tnote{1} & 79 \\ 
  MC+ & 100\tnote{1} & 8 & 37 & 96 & 1 & 85 & 96 & 5 & 58 & 52 & 1 & 74 \\ 
 \hline
 & \multicolumn{12}{c}{Scenario~\ref{sce1b}} \\
\hline
  Lasso & 100\tnote{1} & 45 & 10 & 91 & 11 & 52 & 100\tnote{1} & 96 & 4 & 20 & 6 & 22 \\ 
  Enet & 100\tnote{1} & 63 & 7 & 87 & 18 & 31 & 100\tnote{1} & 98 & 4 & 15 & 6 & 14 \\ 
  AL & 100\tnote{1} & 26 & 19 & 95 & 5 & 81 & 99 & 85 & 5 & 26 & 5 & 35 \\ 
  Aenet & 100\tnote{1} & 30 & 15 & 95 & 6 & 74 & 100\tnote{1} & 87 & 5 & 24 & 5 & 30 \\ 
  SCAD & 100\tnote{1} & 26 & 18 & 93 & 5 & 76 & 100\tnote{1} & 76 & 5 & 23 & 4 & 28 \\ 
  MC+ & 100\tnote{1} & 26 & 17 & 93 & 5 & 76  & 99 & 76 & 5 & 22 & 5 & 27 \\ 
\hline
& \multicolumn{12}{c}{Scenario~\ref{sce2a}} \\
\hline
   Lasso & 98 & 93 & 10 & 84 & 73 & 14 & 98 & 96 & 10 & 47 & 35 & 16\\ 
  Enet & 99 & 98 & 10 & 85 & 80 & 11 & 99 & 98 & 10 & 46 & 38 & 12 \\ 
  AL & 95 & 49 & 18 & 83 & 35 & 27 & 95 & 64 & 15 & 50 & 23 & 28\\ 
  Aenet & 96 & 52 & 17 & 84 & 40 & 21 & 96 & 68 & 14 & 48 & 26 & 20 \\ 
  SCAD & 86 & 74 & 13 & 65 & 45 & 36 & 75 & 60 & 21 & 39 & 26 & 30 \\ 
  MC+ & 87 & 78 & 13 & 65 & 47 & 35 & 73 & 60 & 22 & 39 & 26 & 30 \\ 
\hline
& \multicolumn{12}{c}{Scenario~\ref{sce2b}} \\
\hline
   Lasso & 80 & 64 & 13 & 75 & 60 & 12 & 78 & 69 & 11 & 64 & 57 & 9 \\ 
  Enet & 85 & 73 & 12 & 82 & 69 & 11 & 84 & 79 & 11 & 68 & 64 & 8 \\ 
  AL & 56 & 26 & 19 & 54 & 25 & 20 & 59 & 35 & 17 & 48 & 30 & 13 \\ 
  Aenet & 59 & 30 & 18 & 57 & 29 & 18 & 64 & 43 & 15 & 52 & 36 & 11 \\ 
  SCAD & 43 & 21 & 20 & 42 & 20 & 23 & 46 & 24 & 27 & 41 & 25 & 16 \\ 
  MC+ & 44 & 21 & 20 & 43 & 20 & 23 & 46 & 24 & 26 & 41 & 26 & 16 \\
\hline
\end{tabular}}
\begin{tablenotes}
\item[1] Approximate value
\end{tablenotes}
\end{threeparttable}
\end{table}

With these scenarios, we simulate 2000 spatial point patterns from a Thomas point process using the $\mathtt{rThomas}$ function in the $\mathtt{spatstat}$ package. We set the interaction parameter $\kappa$ to be $\kappa=5 \times 10^{-4}, \kappa=5 \times 10^{-5}$ and let $\omega=20$. Briefly, smaller values of $\omega$ correspond to tighter clusters, and smaller values of $\kappa$ correspond to a fewer number of parents \citep[see e.g.][for further details regarding the  Thomas point process]{moller2003statistical}. For each scenario with different $\kappa$, we fit the intensity to the simulated point pattern realizations.

We report the performances of our estimates in terms of two characteristics: selection and prediction properties. We present the selection properties in Table~\ref{ch2:table:sel} and the prediction properties in Table~\ref{ch2:table:pred}

\setlength{\tabcolsep}{2.5pt}
\renewcommand{\arraystretch}{1.5}
\begin{table}[!ht]
\caption{Empirical prediction properties (Bias, SD, and RMSE) based on 2000 replications of Thomas processes on the domain $D$ for two different values of $\kappa$ and for the two different scenarios. Different penalty functions are considered as well as two estimating equations, the regularized Poisson likelihood (PL) and the regularized weighted Poisson likelihood (WPL).}
\label{ch2:table:pred}
\centering
\scalebox{0.8}{
\begin{tabular}{ l | ccc | ccc | ccc | ccc}
\hline
\multirow{3}{*}{Method}  & \multicolumn{3}{c}{Regularized PL} & \multicolumn{3}{c}{Regularized WPL} & \multicolumn{3}{c}{Regularized PL} & \multicolumn{3}{c}{Regularized WPL} \\
\cline{2-13}
 & \multicolumn{6}{c}{$\kappa=5 \times 10^{-4}$} & \multicolumn{6}{c}{$\kappa=5 \times 10^{-5}$} \\
\cline{2-13}
 & Bias & SD & RMSE & Bias & SD & RMSE & Bias & SD & RMSE & Bias & SD & RMSE \\ 
\cline{2-13}
 & \multicolumn{12}{c}{Scenario~\ref{sce1a}} \\
\hline
  Lasso & 0.19 & 0.19 & 0.27 & 0.43 & 0.29 & 0.52 & 0.29 & 0.60 & 0.67 & 0.94 & 0.53 & 1.08 \\ 
  Enet & 0.27 & 0.22 & 0.35 & 0.72 & 0.32 & 0.79 & 0.34 & 0.66 & 0.74 & 1.21 & 0.40 & 1.27 \\ 
  AL & 0.05 & 0.18 & 0.19 & 0.14 & 0.24 & 0.28 & 0.19 & 0.60 & 0.63 & 0.57 & 0.57 & 0.81 \\ 
  Aenet & 0.07 & 0.19 & 0.20 & 0.20 & 0.27 & 0.33 & 0.22 & 0.60 & 0.64 & 0.69 & 0.55 & 0.88 \\ 
  SCAD & 0.19 & 0.19 & 0.27 & 0.29 & 0.32 & 0.43 & 0.14 & 0.55 & 0.57 & 1.10 & 0.71 & 1.31 \\ 
  MC+ & 0.20 & 0.19 & 0.28 & 0.32 & 0.37 & 0.49 & 0.15 & 0.55 & 0.57 & 1.15 & 0.72 & 1.35 \\ 
 \hline
 & \multicolumn{12}{c}{Scenario~\ref{sce1b}} \\
\hline
  Lasso & 0.18 & 1.03 & 1.05 & 0.57 & 0.58 & 0.81 & 1.97 & 8.00 & 8.23 & 1.85 & 2.11 & 2.81 \\ 
  Enet & 0.27 & 1.32 & 1.34 & 0.81 & 0.73 & 1.09 & 1.87 & 7.73 & 7.96 & 1.94 & 2.02 & 2.80 \\ 
  AL & 0.18 & 0.73 & 0.76 & 0.28 & 0.43 & 0.51 & 1.26 & 6.23 & 6.36 & 1.68 & 1.70 & 2.39 \\ 
  Aenet & 0.21 & 0.72 & 0.75 & 0.36 & 0.44 & 0.57 & 1.05 & 5.45 & 5.55 & 1.76 & 1.49 & 2.31 \\ 
  SCAD & 0.26 & 0.99 & 1.02 & 0.39 & 0.63 & 0.74 & 1.20 & 5.55 & 5.68 & 1.71 & 1.59 & 2.34 \\ 
  MC+ & 0.26 & 0.99 & 1.03 & 0.40 & 0.64 & 0.76 & 1.21 & 5.53 & 5.66 & 1.71 & 1.59 & 2.33 \\ 
   \hline
& \multicolumn{12}{c}{Scenario~\ref{sce2a}} \\
\hline
   Lasso & 1.45 & 1.89 & 2.38 & 2.24 & 2.47 & 3.34 & 0.94 & 8.86 & 8.91 & 4.53 & 5.79 & 7.35 \\ 
  Enet & 1.54 & 1.89 & 2.44 & 2.38 & 2.62 & 3.54 & 1.27 & 6.54 & 6.66 & 4.95 & 4.85 & 6.93\\ 
  AL & 1.57 & 1.80 & 2.39 & 2.20 & 2.16 & 3.09 & 1.33 & 6.38 & 6.52 & 4.31 & 4.50 & 6.23 \\ 
  Aenet & 2.05 & 1.60 & 2.59 & 2.64 & 2.11 & 3.38 & 1.95 & 4.75 & 5.13 & 4.89 & 3.73 & 6.14  \\ 
  SCAD & 2.26 & 1.75 & 2.86  & 3.84 & 2.43 & 4.54 & 3.74 & 3.45 & 5.09 & 5.79 & 2.73 & 6.40 \\ 
   MC+ & 2.45 & 1.77 & 3.02& 3.95 & 2.39 & 4.61 & 3.81 & 3.41 & 5.12 & 5.82 & 2.71 & 6.42 \\ 
\hline
& \multicolumn{12}{c}{Scenario~\ref{sce2b}} \\
\hline
   Lasso & 3.28 & 2.87 & 4.36 & 3.36 & 3.20 & 4.64 & 3.85 & 13.41 & 13.95 & 4.61 & 11.20 & 12.11 \\ 
  Enet & 3.39 & 2.45 & 4.18 & 3.48 & 2.75 & 4.44 & 3.76 & 7.86 & 8.71 & 4.66 & 6.96 & 8.37 \\ 
  AL & 3.64 & 1.59 & 3.97 & 3.69 & 1.78 & 4.10 & 3.89 & 8.99 & 9.80 & 4.70 & 6.95 & 8.39 \\ 
  Aenet & 3.71 & 1.34 & 3.95 & 3.79 & 1.58 & 4.10 & 4.03 & 4.89 & 6.34 & 4.88 & 4.38 & 6.55 \\ 
  SCAD & 4.56 & 2.22 & 5.07 & 4.67 & 2.27 & 5.19 & 5.22 & 3.27 & 6.16 & 5.65 & 3.18 & 6.48 \\ 
  MC+ & 4.53 & 2.24 & 5.05 & 4.64 & 2.29 & 5.18 & 5.23 & 3.25 & 6.15 & 5.66 & 3.21 & 6.51 \\ 
   \hline

\end{tabular}}
\end{table}

To evaluate the selection properties of the estimates, we consider the true positive rate (TPR), the false positive rate (FPR), and the positive predictive value (PPV). We want to find the methods which have a TPR close to 100$\%$ meaning that it can select correctly all the true covariates, a FPR close to 0 showing that it can remove all the extra covariates from the model, and a PPV close to 100$\%$ indicating that, for Scenario~\ref{ch2:sce1} (resp. Scenario~\ref{ch2:sce2}), it can keep exactly the two (resp. five) true covariates and remove all the 48 (resp. 45) extra covariates. In general, for both regularized PL and regularized WPL, the best selection properties are obtained from larger $\kappa$ $(5 \times 10^{-4})$ which indicates weaker spatial dependence. To compare the regularization methods, we emphasize here that the main difference between regularization methods which satisfy (adaptive lasso, adaptive elastic net, SCAD, and MC+) and which cannot satisfy (lasso, elastic net) our theorems is that the methods which cannot satisfy our theorems tend to over-select covariates, leading to suffering from larger FPR and smaller PPV in general. Among all regularization methods considered in this study, adaptive lasso and adaptive elastic net seem to outperform the other methods in most cases. Although adaptive lasso and adaptive elastic net perform quite similarly, the adaptive lasso is slightly better.

In this simulation study, we are still able to show that even when the strong multicollinearity exists such as in Scenario~\ref{sce1a}, our proposed methods work well for the penalization methods satisfying our theorems. However, as probably expected, our methods are getting difficult to distinguish between the important and the noisy covariates as the setting becomes more and more complex. In the experiments we conduct, we find that the regularized PL and WPL (with adaptive lasso) perform quite similar for the easiest (Scenario~\ref{sce1a}) and the toughest (Scenario~\ref{sce2b}) setting. For Scenarios~\ref{sce1b} and \ref{sce2a}, the regularized WPL with adaptive lasso seems to be more favorable. From Table~\ref{ch2:table:sel}, we would recommend in general to combine the regularized WPL with the adaptive lasso to perform variable selection.

Table~\ref{ch2:table:pred} gives the prediction properties of the estimates (except for $\beta_0$ which is excluded) in terms of biases, standard deviations (SD), and square root of mean squared errors (RMSE), some criteria we define by
\begin{align*}
\mathrm{Bias}&=\left [ {\sum_{j=1}^{50} { \{\hat{\mathbb{E}}(\hat \beta_j)-\beta_{0j}\}^2}}  \right ]^\frac{1}{2},
\mathrm{SD}=\left [ {\sum_{j=1}^{50} { \hat \sigma_j^2}}  \right ]^\frac{1}{2},
\mathrm{RMSE}=\left [ {\sum_{j=1}^{50} { \hat{\mathbb{E}}(\hat \beta_j-\beta_{0j})^2}}  \right ]^\frac{1}{2},
\end{align*}
where $\hat{\mathbb{E}}(\hat \beta_j)$ and $ \hat \sigma_j^2$ are respectively the empirical mean and variance of the estimates $\hat \beta_j$, for $j=1,\ldots,50$.

In general, the properties improve with larger $\kappa$ due to weaker spatial dependence. Regarding the regularization methods considered in this study, adaptive lasso and adaptive elastic net perform best. Adaptive elastic net becomes more preferable than adaptive lasso for a clustered process ($\kappa= 5 \times 10^{-5}$) and for a structured spatial data (Scenarios~\ref{sce1b} and \ref{sce2b}). The  adaptive elastic net is more efficient than the adaptive lasso especially in the complex situation: large number of covariates, strong multicollinearity, clustered processes, and complex spatial structure due to the advantage of combining $l_1$ and $l_2$ penalties.

By employing regularized WPL, we have potentially more efficient estimates than that of the regularized PL, especially for the more clustered process. However, this does not mean that the regularized WPL is able to improve the RMSE since it usually introduces extra biases. Regularized WPL seems more appropriate for the case having covariates with complex spatial structure (Scenarios~\ref{sce1b} and \ref{sce2b}). Otherwise, regularized PL is more favorable. From Table~\ref{ch2:table:pred}, when the focus is on prediction, we would recommend to apply adaptive elastic net as a general advice, and we would combine with regularized WPL if the covariates have complex spatial structure (e.g. Scenarios~\ref{sce1b} and \ref{sce2b}) or combine with regularized PL if there is no evidence of complex spatial structure in the covariates (e.g. Scenarios~\ref{sce1a} and \ref{sce2a}).

Note that, from Table~\ref{ch2:table:sel}, the adaptive lasso is more preferable if the focus is on variable selection while, from Table~\ref{ch2:table:pred}, the adaptive elastic net is more favorable if the focus is for prediction. To have a more general recommendation, we would recommend applying adaptive elastic net when we are faced with a complex situation: a large number of covariates, strong multicollinearity, clustered processes and complex spatial structure. By combining $l_1$ and $l_2$ penalties, the adaptive elastic net provides a nice balance between selection and  prediction properties. This is why in most complex cases (Scenario~\ref{ch2:sce2} with $\kappa=5 \times 10^{-5}$), adaptive elastic net decides to choose more covariates than adaptive lasso (which includes true and noisy covariates) to suffer from slightly less appropriate properties for the selection performance but to be able to improve significantly the prediction properties.

\subsection{Application to forestry datasets}
\label{ch2:sec8}
We now consider the study of ecology in a tropical rainforest in Barro Corrolado Island (BCI), Panama, described previously in Section~\ref{ch2:intro}. In particular, we are interested  in studying the spatial distribution of 3,604 locations of {\it Beilschmiedia pendula Lauraceae} (BPL) trees by modeling its intensity as a log-linear function of 93 covariates consisting of 2 topological attributes, 13 soil properties, and 78 interactions between two soil nutrients.Regarding the relatively large number of covariates, we apply our proposed methods to select few covariates among them and estimate their coefficients. In particular, we use the regularized Poisson methods with the lasso, adaptive lasso, and adaptive elastic net. Note that we center and scale all the covariates to observe which covariates owing relatively large effect on the intensity.

\setlength{\tabcolsep}{2.5pt}
\renewcommand{\arraystretch}{1.5}
\begin{table}[!ht]
\caption{Number of selected and non-selected covariates among 93 covariates by regularized Poisson likelihood with lasso, adaptive lasso and adaptive elastic net regularization.}
\label{numb}
\centering
\begin{tabular}{l cccc}
  \hline
\hline
\multirow{2}{*}{Method}  & \multicolumn{2}{c}{Regularized PL} & \multicolumn{2}{c}{Regularized WPL}\\
\cline{2-5}
 & \#Selected & \#Non-selected & \#Selected & \#Non-selected \\ 
  \hline
\hline
LASSO & 77 & 16 & 20 & 73 \\ 
  AL & 50 & 43 & 9 & 84 \\ 
  AENET & 69 & 24 & 9 & 84 \\ 
   \hline
\end{tabular}
\end{table}

\setlength{\tabcolsep}{2.5pt}
\renewcommand{\arraystretch}{1.5}
\begin{table}[ht]
\caption{Nine common covariates selected}
\label{common}
\centering
\begin{tabular}{l | rrr | rrr}
\hline
\hline
\multirow{2}{*}{Covariates} & \multicolumn{3}{c}{Regularized PL} & \multicolumn{3}{c}{Regularized WPL} \\ 
\cline{2-7}
 & LASSO & AL & AENET & LASSO & AL & AENET \\ 
\hline
\hline
Elev & 0.33 & 0.37 & 0.34 & 0.23 & 0.14 & 0.14 \\ 
  Slope & 0.37 & 0.37 & 0.37 & 0.45 & 0.44 & 0.46 \\ 
  Cu & 0.45 & 0.30 & 0.30 & 0.16 & 0.22 & 0.19 \\ 
  Mn & 0.11 & 0.10 & 0.11 & 0.18 & 0.14 & 0.14 \\ 
  P & -0.49 & -0.45 & -0.48 & -0.50 & -0.43 & -0.39 \\ 
  Zn & -0.69 & -0.54 & -0.70 & -0.21 & -0.31 & -0.25 \\ 
  Al:P & -0.28 & -0.24 & -0.28 & -0.13 & -0.14 & -0.13 \\ 
  Mg:P & 0.49 & 0.26 & 0.30 & 0.38 & 0.38 & 0.34 \\ 
  N.Min:pH & 0.42 & 0.39 & 0.39 & 0.22 & 0.17 & 0.17 \\ 

   \hline
\end{tabular}
\end{table}

We present in Table~\ref{numb} the number of selected and non-selected covariates by each method. Out of 93 covariates, more than $50\%$ from the total number of covariates are selected by regularized PL while much fewer covariates are selected by regularized WPL. The regularized PL seems to overfit the model. 

Regarding lasso method, 77 covariates are selected by regularized PL method while 20 covariates are selected by regularized WPL. Compared to the two adaptive methods (i.e., adaptive lasso and adaptive elastic net), lasso tends to keep less important covariates. This may explain why lasso cannot satisfy our Theorem~\ref{ch2:the2}. In terms of selection properties, adaptive lasso and adaptive elastic net perform similarly when regularized WPL is applied.

Table~\ref{common} gives the information regarding nine covariates commonly selected among the six methods. Although the magnitudes of the estimates can be slightly different, the signs all agree with each other.

These results suggest that BPL trees favor to live in the areas of higher elevation and slope with a high concentration of Copper and Manganese in the soil. Furthermore, BPL trees prefer to live in the areas with lower concentration levels of Phosphorus and Zinc in the soil. The interaction between Aluminum and Phosphorus gives a negative association with the appearance of BPL trees while the interaction between Magnesium and Phosphorus and the interaction between Nitrogen mineralization and pH show a positive association with the occurrence of BPL trees. The maps of 3,604 locations of BPL trees, as well as the nine commonly selected covariates, are depicted in Figure~\ref{common var}.

\begin{figure}[!ht]

\renewcommand{\arraystretch}{0}
\setlength{\tabcolsep}{1pt}
 \begin{tabular}{l ccc}  
 \begin{tabular}{l}  
           \includegraphics[width=0.35\textwidth]{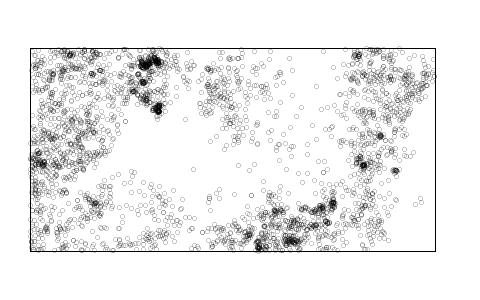}
           \end{tabular}

\begin{tabular}{l l l}
\includegraphics[width=0.21\textwidth]{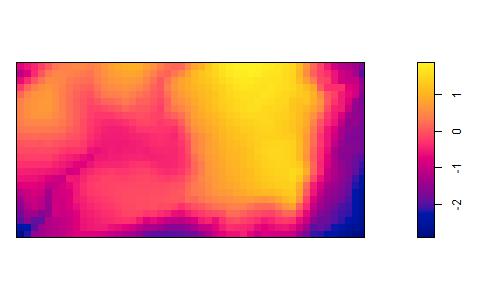} & \includegraphics[width=0.21\textwidth]{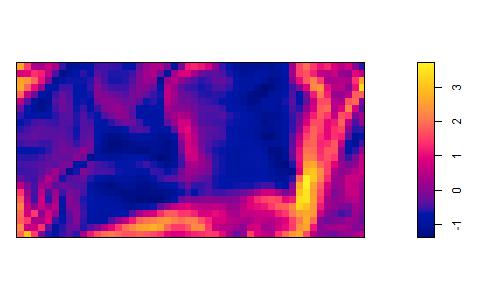} &  \includegraphics[width=0.21\textwidth]{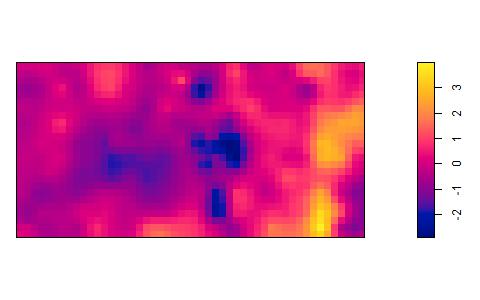} \\
\includegraphics[width=0.21\textwidth]{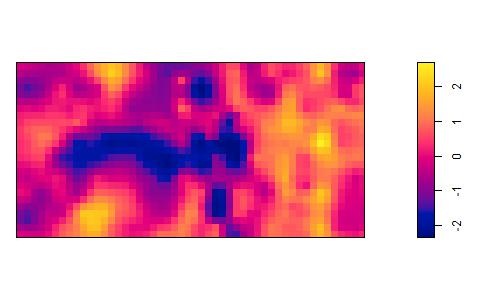} & \includegraphics[width=0.21\textwidth]{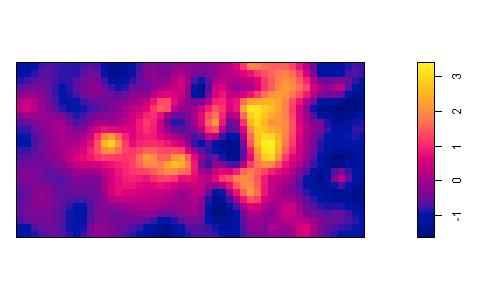} &  \includegraphics[width=0.21\textwidth]{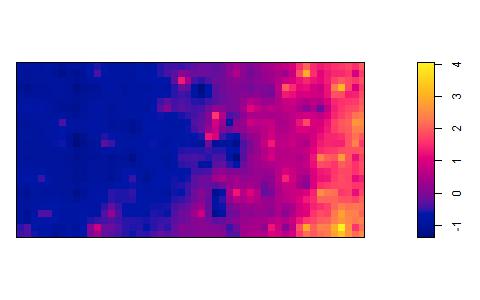} \\
\includegraphics[width=0.21\textwidth]{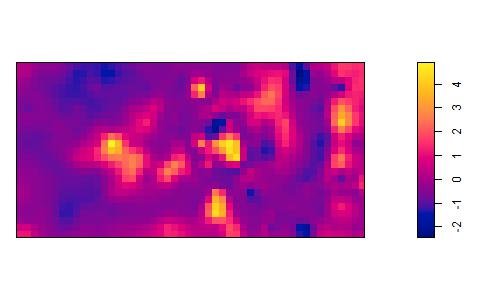} &  \includegraphics[width=0.21\textwidth]{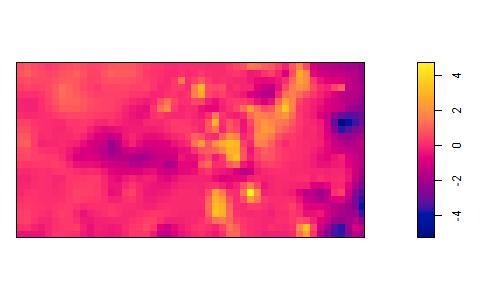} &  \includegraphics[width=0.21\textwidth]{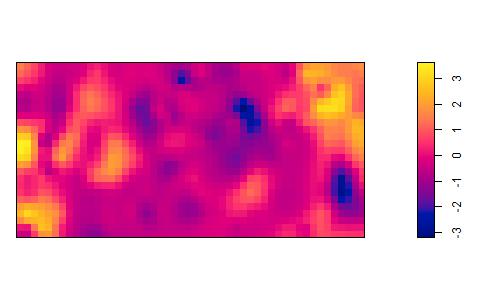} \\
\end{tabular}
\end{tabular}
\caption{Maps of 3,604 locations of BPL trees and the nine common selected covariates, from left to right, row~1:  elevation, slope and Copper, row~2: Manganese, Phosphorus and Zinc, row~3: the interaction between Aluminum and Phosphorus, between Magnesium and Phosphorus, and between Nitrogen mineralisation and pH.}
\label{common var}
\end{figure}



\section{Conclusion and discussion}
\label{conclu}
We consider feature selection techniques for spatial point processes intensity estimation by regularizing  estimating equations derived from Poisson and logistic regression likelihoods in a setting where the number of parameters diverges as the volume of observation domain increases. Under some conditions, we prove that the estimates obtained from such setting satisfy consistency, sparsity, and asymptotic normality. Our results are available for large classes of spatial point processes and for many penalty functions.

We conduct  simulation experiments to evaluate the finite sample properties of the regularized Poisson estimator and regularized weighted Poisson estimator. From the results, we would recommend in general the combination between regularized WPL and adaptive lasso if the concern is on variable selection. Furthermore, when the focus is on prediction, the regularized WPL combined with the adaptive elastic net is more preferable for the situation where there is a complex spatial structure in the covariates. For more general advice, we would recommend using the adaptive elastic net rather than the adaptive lasso since the adaptive elastic net is able to balance the selection and the prediction properties by combining the $l_1$ and the $l_2$ penalties.


To implement our methods, we combine the $\texttt{spatstat}$ $\texttt{R}$ package and the two $\texttt{R}$ packages $\texttt{glmnet}$ and $\texttt{ncvreg}$ dealing with penalized generalized linear models. This results in a computationally fast procedure even when the number of covariates is large. It is worth noticing that, as other regularization methods, our methods also rely on the selection of the tuning parameter. As the study in a classical regression analysis, the BIC-type methods are proposed to obtain selection consistent estimator \citep[see e.g.][]{zou2007degrees, wang2007tuning, wang2009shrinkage}. We have numerical evidence from simulation studies that this criterion can satisfy the selection consistency.  Theoretical justification in this spatial point process framework is the purpose of a future research.

We apply our methods to the Barro Corrolado Island study to estimate the intensity of {\it Beilschmiedia pendula Lauraceae} (BPL) tree as a log-linear function of 93 environmental covariates. Regularized weighted Poisson likelihood combined with adaptive elastic net performs similarly to adaptive lasso. Among 93 covariates, we find nine spatial covariates which may have a high influence to the appearance of BPL trees, including two topological attributes: elevation and slope, four soil nutrients: Copper, Manganese, Phosphorus and Zinc, and three interaction between two soil properties: the interaction between  Aluminum and Phosphorus, between Magnesium and Phosphorus, and between Nitrogen mineralisation and pH.

A further work would consider to include the 296 other species of trees, which were surveyed in the same observation region, to study the existence of any competition between BPL and other species of trees in the forest. In such a situation, the methods used in this study may face some computational issues. The Dantzig selector \citep{candes2007dantzig} might be a good alternative since the implementation for linear models (and generalized linear models) results in a linear programming. Thus, more competitive algorithms are available. It would be interesting to bring this approach to spatial point process framework.

\vskip 14pt
\noindent {\bf Acknowledgements}
\vskip 4pt
We thank A. L. Thurman who kindly shared the $\texttt{R}$ code used for the simulation study in \cite{thurman2015regularized} and P. Breheny who kindly provided his code used in $\texttt{ncvreg}$ $\texttt{R}$ package. We also thank R. Drouilhet for technical help. The research of A. Choiruddin is supported by The Danish Council for Independent Research -- Natural Sciences, grant DFF -- 7014-00074 "Statistics for point processes in space and beyond", and by the ”Centre for Stochastic Geometry and Advanced Bioimaging”, funded by grant 8721 from the Villum Foundation. The research of J.-F. Coeurjolly is supported by the Natural Sciences and Engineering Research Council of Canada. The research of F. Letu\'e is supported by ANR-11-LABX-0025 Persyval-lab (project Persyvact2).

The BCI soils data sets were collected and analyzed by J. Dalling, R. John, K. Harms, R. Stallard and J. Yavitt with support from NSF DEB021104,021115, 0212284,0212818 and OISE 0314581, and STRI Soils Initiative and CTFS and assistance from P. Segre and J. Trani. Datasets are available at the CTFS website \\ $\texttt{http://ctfs.si.edu/webatlas/datasets/bci/soilmaps/BCIsoil.html}$.
\par

\bibliographystyle{plainnat}
\bibliography{refSS}

\appendix


\section{Auxiliary Lemma}  \label{ch2:auxLemma}

The following lemma is used in the proof of Theorem~\ref{ch2:the1} and Lemma~\ref{ch2:sparsity} (which includes Lemma~\ref{ch2:clt} and Theorem~\ref{ch2:the2}). Throughout the proofs, the notation $\mathbf X_n = O_{\mathrm P} (x_n)$ or $\mathbf X_n = o_{\mathrm P} (x_n)$ for a random vector $\mathbf X_n$ and a sequence of real numbers $x_n$ means that $\|\mathbf X_n\|=O_{\mathrm P}(x_n)$ and $\|\mathbf X_n\|=o_{\mathrm P}(x_n)$. In the same way for a vector $\mathbf V_n$ or a squared matrix $\mathbf M_n$, the notation $\mathbf V_n=O(x_n)$ and  $\mathbf M_n=O(x_n)$ mean that $\|\mathbf V_n\|=O(x_n)$ and  $\|\mathbf M_n\|=O(x_n)$.

\begin{lemma} \label{ch2:bound}
Under conditions ($\mathcal C$.\ref{ch2:C:Dn})-($\mathcal C$.\ref{ch2:C:mixing}), the following result holds as \linebreak $n\to \infty$
	\begin{align}
 \ell^{(1)}_n(w;\boldsymbol \beta_0)  =O_\mathrm{P} \left( \sqrt { p_n |D_n| } \right)  \label{ch2:eq:ln}.	
\end{align}
\end{lemma}

\begin{proof}
Using Campbell Theorems~\eqref{ch2: eq:campbell}-\eqref{ch2:eq:campbell2}, the score vector $ \ell^{(1)}_n(w;\boldsymbol \beta_0)  $ has variance
\begin{align*}
\mathrm{Var}[ \ell^{(1)}_{n}(w; \boldsymbol \beta_0)]= \mathbf{B}_{n}(w; \boldsymbol \beta_{0})+\mathbf{C}_{n}(w; \boldsymbol \beta_{0}).
\end{align*}
Conditions ($\mathcal C$.\ref{ch2:C:rhok})-($\mathcal C$.\ref{ch2:C:mixing}) allow us to obtain that $\sup_{u \in \mathbb{R}^d} \int_{\mathbb{R}^d} \{g(u,v)-1\} \mathrm{d}v < \nolinebreak \infty$. We then deduce using conditions ($\mathcal C$.\ref{ch2:C:Dn})-($\mathcal C$.\ref{ch2:C:cov}) that
\begin{align*}
\mathbf{B}_{n}(w; \boldsymbol \beta_{0})+\mathbf{C}_{n}(w; \boldsymbol \beta_{0})=O(p_n |D_n|).
\end{align*}
The result is proved since for any centered real-valued stochastic process ${Y_n}$ with finite variance $\mathrm{Var}[{Y_n}]$, ${Y_n}=O_\mathrm{P}(\sqrt{\mathrm{Var}[{Y_n}]})$.
\end{proof}


\section{Proof of Theorem~\ref{ch2:the1}} \label{ch2:proof1}

In the proof of this result and the following ones, the notation $\kappa$ stands for a generic constant which may vary from line to line. In particular this constant is independent of $n$, $\boldsymbol \beta_0$ and $\mathbf k$.

\begin{proof}
Let $ d_n = \sqrt p_n(|D_n|^{-1/2}+a_n)$, and $\mathbf{k}=\{k_1, k_2, \ldots, k_{p_n}\}^\top $. We remind the reader that the estimate of $\boldsymbol\beta_0$ is defined as the maximum of the function $Q_n$ (given by~\eqref{Qn}) over $\Theta$, an open convex bounded set of $\mathbb R^{p_n}$ for any $n \geq 1$. For any $\mathbf k$ such that $\|\mathbf k\|\leq K<\infty$, $\boldsymbol \beta_0 + d_n \mathbf k \in \Theta$ for $n$  sufficiently large. Assume this is valid in the following.
To prove Theorem~\ref{ch2:the1}, we aim at proving that for any given $\epsilon>0$, there exists sufficiently large $K>0$ such that for $n$ sufficiently large 
\begin{equation}
\label{ch2:eq:15}
\mathrm{P}\bigg(\sup_{\|\mathbf{k}\| = K} \Delta_n(\mathbf k)>0\bigg)\leq \epsilon,
\quad \mbox{ where } \Delta_n(\mathbf k) = Q_n(w;\boldsymbol \beta_0+d_n\mathbf{k})-Q_n(w;\boldsymbol \beta_0).
\end{equation}
Equation~\eqref{ch2:eq:15} will imply that with probability at least $1-\epsilon$, there exists a local maximum in the ball $\{\boldsymbol \beta_0+d_n\mathbf{k}:\|\mathbf{k}\| \leq K\}$, and therefore  a local maximizer $\boldsymbol{\hat{\beta}}$ is such that $\|{ \boldsymbol {\hat \beta}-\boldsymbol \beta_0}\|=O_\mathrm{P}(d_n)$.  We decompose $\Delta_n(\mathbf k)$ as $\Delta_n(\mathbf k)= T_1+T_2$ where
\begin{align*}
	T_1 & = \ell_n(w;\boldsymbol \beta_0+d_n\mathbf{k})-\ell_n(w; \boldsymbol \beta_0) \\
	T_2 & = |D_n|{\sum_{j=1}^{p_n} \big( p_{\lambda_{n,j}}(|\beta_{0j}|)}- p_{\lambda_{n,j}}(|\beta_{0j}+d_nk_j|)\big).
\end{align*}
Since $\rho(u;\cdot)$ is infinitely continuously differentiable and $\ell_n^{(2)}(w;\boldsymbol\beta) =-\mathbf A_n(w;\boldsymbol\beta)$, then using a second-order Taylor expansion there exists $t\in (0,1)$ such that
\begin{align*}
	T_1 =& \, d_n \mathbf k^\top \ell_n^{(1)}(w;\boldsymbol \beta_0) - \frac12d_n^2\mathbf k^\top \mathbf{A}_n(w;\boldsymbol \beta_0) \mathbf k  \\
	&+ \frac12d_n^2\mathbf k^\top \left( \mathbf{A}_n(w;\boldsymbol \beta_0) -\mathbf{A}_n(w;\boldsymbol \beta_0 + td_n \mathbf k) \right) \mathbf k .	
\end{align*}
By conditions~($\mathcal C$.\ref{ch2:C:Theta})-($\mathcal C$.\ref{ch2:C:cov}), there exists a non-negative constant $\kappa$ such that
\[
	\frac12\|\mathbf{A}_n(w;\boldsymbol \beta_0) -\mathbf{A}_n(w;\boldsymbol \beta_0 + td_n \mathbf k) \| \leq  \kappa d_n |D_n| p_n.
\]
Now, denote $\check \nu := \liminf_{n\to \infty} \nu_{\min}(|D_n|^{-1}\mathbf{A}_n(w;\boldsymbol \beta_0))$. By condition ($\mathcal C$.\ref{ch2:C:An}), we have  that for any $\mathbf k$
\[
	 0<\check \nu \leq  \frac{\mathbf k^\top \left( |D_n|^{-1}\mathbf A_n(w;\boldsymbol\beta_0)\right) \mathbf k}{\|\mathbf k\|^2}.
\]
Therefore, we have
\[
	T_1 \leq d_n \|\ell_n^{(1)}(w;\boldsymbol \beta_0)\| \, \| \mathbf k \|  - \frac{\check \nu}2 d_n^2 |D_n| \|\mathbf k\|^2 +\kappa p_n d_n^3  |D_n|\|\mathbf k\|^2.
\]
Now by the condition ($\mathcal C$.\ref{ch2:C9}) and by assumption that $a_n=O(|D_n|^{-1/2})$, we obtain $p_n d_n  = o(1)$, so $\kappa p_n d_n^3  |D_n|\|\mathbf k\|^2 = o(1) d_n^2 |D_n|\|\mathbf k\|^2$. Hence, for $n$ sufficiently large
\[
	T_1 \leq d_n \|\ell_n^{(1)}(w;\boldsymbol \beta_0)\| \, \| \mathbf k \|  - \frac{\check \nu}4 d_n^2 |D_n| \|\mathbf k\|^2.
\]

Regarding the term $T_2$,
\[
T_2\leq T_2^\prime := |D_n|{\sum_{j=1}^{s} \big( p_{\lambda_{n,j}}(|\beta_{0j}|)}- p_{\lambda_{n,j}}(|\beta_{0j}+d_nk_j|)\big) 	
\]
since for any $j$ the penalty function $p_{\lambda_{n,j}}$ is non-negative and $p_{\lambda_{n,j}}(|\beta_{0j}|)=0$ for $j=s+1,\dots,p_n$.

From ($\mathcal C$.\ref{ch2:C:plambda}), for $n$ sufficiently large, $p_{\lambda_{n,j}}$ is twice continuously differentiable for every $\beta_j = \beta_{0j}+t d_n k_j$ with $t\in (0,1)$. Therefore using a third-order Taylor expansion, there exist $t_j \in (0,1)$, $j=1,\dots,s$ such that $-T_2^\prime= T_{2,1}^\prime + T_{2,2}^\prime + T_{2,3}^\prime$, where
\begin{align*}
T_{2,1}^\prime &=d_n|D_n|\sum_{j=1}^{s} k_j p_{\lambda_{n,j}}^\prime(|\beta_{0j}|) \sign(\beta_{0,j})  \leq  \sqrt s a_n d_n |D_n| \, \|\mathbf k\| \leq  d_n^2 |D_n| \, \| \mathbf k\|, \\
T_{2,2}^\prime &= \frac12 d_n^2 |D_n|\sum_{j=1}^{s} k_j^2 p^{\prime\prime}_{\lambda_{n,j}}  (|\beta_{0j}|) \leq  c_n d_n^2|D_n|  \|\mathbf k\|^2, \\
T_{2,3}^\prime &= \frac16 d_n^3|D_n| \sum_{j=1}^{s} k_j^3 p^{\prime\prime\prime}_{\lambda_{n,j}}  (|\beta_{0j}+t_j d_n k_j|) \leq \kappa d_n^3 |D_n|.
\end{align*}
The three inequalities above are obtained using the definitions of $a_n$ and $c_n$, condition~($\mathcal C$.\ref{ch2:C:plambda}) and Cauchy-Schwarz inequality. We deduce that for $n$ sufficiently large
\begin{align*}
	T_2 \leq | T_2^\prime| &\leq 2 d_n^2 |D_n| \|\mathbf k\|,
\end{align*}
and then
\[
	\Delta_n(\mathbf k) \leq d_n \|\ell_n^{(1)}(w;\boldsymbol \beta_0)\| \, \| \mathbf k \| -\frac{\check \nu}4 d_n^2 |D_n| \|\mathbf k\|^2 + 2 d_n^2 |D_n| \|\mathbf k\|.
\]
We now return to (\ref{ch2:eq:15}): for $n$ sufficiently large
\[
\mathrm{P}\bigg({\sup_{\|\mathbf{k}\|= K}  \Delta_n(\mathbf{k})>0}\bigg) \leq 	\mathrm P \bigg(
\| \ell_n^{(1)}(w;\boldsymbol\beta_0)\| > \frac{\check \nu}4 d_n|D_n| K - 2 d_n |D_n|
 \bigg).
\]
Since $d_n |D_n|=O(\sqrt{p_n |D_n|})$, by choosing $K$ large enough, there exists $\kappa$ such that for $n$  sufficiently large
\[
	\mathrm P \bigg( \sup_{\|\mathbf k\|=K}  \Delta_n(\mathbf k) >0\bigg) \leq \mathrm P \bigg( \|\ell_n^{(1)}(w;\boldsymbol \beta_0)\| >\kappa \sqrt{p_n |D_n|}\bigg) \leq \epsilon
\]
for any given $\epsilon>0$ from \eqref{ch2:eq:ln} in Lemma~\ref{ch2:bound}.

\end{proof}


\section{Proof of Theorem~\ref{ch2:the2}} \label{ch2:proof2}
Before proving Theorem~\ref{ch2:the2}, we present Lemmas~\ref{ch2:sparsity}-\ref{ch2:clt}. Lemma~\ref{ch2:sparsity} is used to prove Theorem~\ref{ch2:the2}(i) while  Lemma~\ref{ch2:clt} is used to derive Theorem~\ref{ch2:the2}(ii).
\begin{lemma}
\label{ch2:sparsity}
Assume the conditions ($\mathcal C$.\ref{ch2:C:Dn})-($\mathcal C$.\ref{ch2:C:plambda}) hold. If $a_n=O(|D_n|^{-1/2})$ and $b_n \sqrt{|D_n|/p_n^2} \to \infty$ as $n\to\infty$, then with probability tending to $1$, for any {$\boldsymbol \beta_1$} satisfying $\|{\boldsymbol \beta_1 - \boldsymbol \beta_{01}}\|=O_\mathrm{P}(\sqrt{p_n/|D_n|})$, and for any constant $K_1 > 0$,
\begin{align*}
Q_n\Big(w;({\boldsymbol \beta_1}^\top,\mathbf{0}^\top)^\top \Big)
= \max_{\| \boldsymbol \beta_2\| \leq K_1 \sqrt{p_n/|D_n|}}
Q_n\Big(w;({\boldsymbol \beta_1}^\top,{\boldsymbol \beta_2}^\top)^\top \Big).
\end{align*}
\end{lemma}
\begin{proof}
Let $\varepsilon_n= K_1 \sqrt{p_n/|D_n|}$. It is sufficient to show that with probability tending to $1$ as ${n\to \infty}$, for any ${\boldsymbol \beta_1}$ satisfying $\|{\boldsymbol \beta_1 -\boldsymbol \beta_{01}}\|=O_\mathrm{P}(\sqrt{p_n/|D_n|})$, we have for any $j=s+1, \ldots, p_n$

\begin{equation}
\label{ch2:sparsitya}
\frac {\partial Q_n(w;\boldsymbol \beta)}{\partial\beta_j}<0 \quad
\mbox { for } 0<\beta_j<\varepsilon_n, \mbox{ and}
\end{equation}

\begin{equation}
\label{ch2:sparsityb}
\frac {\partial Q_n(w;\bf \boldsymbol \beta)}{\partial\beta_j}>0 \quad
\mbox { for } -\varepsilon_n<\beta_j<0.
\end{equation}
From (\ref{ln}),
\begin{align*}
\frac {\partial \ell_n(w;\boldsymbol \beta)}{\partial\beta_j} = \frac {\partial \ell_n{(w;\boldsymbol \beta_0)}}{\partial\beta_j} + R_n,
\end{align*}
where $R_n = \int_{D_n} w(u) z_j(u)\big(\rho(u;\boldsymbol \beta)-\rho(u;\boldsymbol \beta_0)\big) \mathrm{d}u$. Using similar arguments used in the proof of Lemma~\ref{ch2:bound}, we can prove that
\begin{align*}
\frac {\partial \ell_n{(w;\boldsymbol \beta_0)}}{\partial\beta_j} = O_\mathrm{P}(\sqrt{|D_n|}).
\end{align*}
Let $u \in \mathbb{R}^d$. By Taylor expansion, there exists $ t\in (0,1), $ such that 
\begin{align*}
\rho(u;\boldsymbol \beta) = \rho(u;\boldsymbol \beta_0) + (\boldsymbol \beta-\boldsymbol \beta_0)^\top \mathbf{z}(u) \rho(u;\boldsymbol \beta_0 + t(\boldsymbol \beta-\boldsymbol \beta_0 )).
\end{align*}
For $n$ sufficiently large, $\boldsymbol \beta_0 + t(\boldsymbol \beta-\boldsymbol \beta_0 ) \in \Xi(\boldsymbol \beta_0)$ defined in condition ($\mathcal C$.\ref{ch2:C:Theta}). Therefore, for $n$ sufficiently large, we have by Cauchy-Schwarz inequality and conditions ($\mathcal C$.\ref{ch2:C:Theta})-($\mathcal C$.\ref{ch2:C:cov})
\begin{align*}
|R_n| \leq \kappa \int_{D_n} \| \boldsymbol \beta-\boldsymbol \beta_0 \| \| \mathbf{z}(u)\| \mathrm{d}u = O_\mathrm{P}(\sqrt{|D_n| p_n^2}).
\end{align*}
We therefore deduce that for any $j=s+1, \dots, p_n$
\begin{align}
\frac {\partial \ell_n(w;\boldsymbol \beta)}{\partial\beta_j} = O_\mathrm{P}(\sqrt{|D_n| p_n ^2} ) \label{ch2:Op}.
\end{align}

Now, we want to prove (\ref{ch2:sparsitya}). Let $0<\beta_j<\varepsilon_n$ and $b_n$ be the sequence given by~(\ref{ch2:eq:bn}). By condition ($\mathcal C$.\ref{ch2:C:plambda}), $b_n$ is well-defined and since by the assumption $b_n \sqrt{|D_n|/p_n^2} \to \infty$, in particular,  $b_n>0$ for $n$ sufficiently large. Therefore, for $n$ sufficiently large,
\begin{align*}
\mathrm{P} \left ( \frac {\partial Q_n(w;\boldsymbol \beta)}{\partial\beta_j}<0 \right)&=\mathrm{P} \left ( \frac {\partial \ell_n(w;\boldsymbol \beta)}{\partial\beta_j} - |D_n|p'_{\lambda_{n,j}}(|\beta_j|)\sign(\beta_j)<0 \right)\\
&=\mathrm{P} \left ( \frac {\partial \ell_n(w;\boldsymbol \beta)}{\partial\beta_j}< |D_n|p'_{\lambda_{n,j}}(|\beta_j|) \right)\\
& \geq \mathrm{P} \left ( \frac {\partial \ell_n(w;\boldsymbol \beta)}{\partial\beta_j}< |D_n|b_n \right)\\
&= \mathrm{P} \left ( \frac {\partial \ell_n(w;\boldsymbol \beta)}{\partial\beta_j}< \sqrt{|D_n|p_n^2} \; \sqrt{\frac{|D_n|}{p_n^2}}b_n \right).
\end{align*}
The assertion (\ref{ch2:sparsitya}) is therefore deduced from (\ref{ch2:Op}) and from the assumption that $b_n \sqrt{|D_n|/p_n^2} \to \infty$ as $n  \to \infty$. We proceed similarly to prove (\ref{ch2:sparsityb}).
\end{proof}


\begin{lemma} \label{ch2:clt}
Under the conditions ($\mathcal C$.\ref{ch2:C:Dn})-($\mathcal C$.\ref{ch2:C:plambda}) and the conditions required in \linebreak Lemma~\ref{ch2:sparsity}, the following convergence holds in distribution as $n\to \infty$
	\begin{align}
\{ \mathbf{B}_{n,11}(w; \boldsymbol \beta_{01})+\mathbf{C}_{n,11}(w; \boldsymbol \beta_{01})\}^{-1/2}\ell^{(1)}_{n,1}(w; \boldsymbol \beta_{01}) \xrightarrow{d} \mathcal{N}(\mathbf{0},\mathbf{I}_{s}) \label{eq:normal},
\end{align}
where $\ell^{(1)}_{n,1}(w;\boldsymbol \beta_{0})$ is the first $s$ components of $\ell^{(1)}_{n}(w;\boldsymbol \beta_{0})$ and $\mathbf{B}_{n,11}(w;\boldsymbol{\beta}_{0})$ \linebreak $($resp. $\mathbf{C}_{n,11}(w;\boldsymbol{\beta}_{0}))$ is the $s \times s$ top-left corner of $\mathbf{B}_{n}(w;\boldsymbol{\beta}_{0})$ $($resp  $\mathbf{C}_{n}(w;\boldsymbol{\beta}_{0}))$.
\end{lemma}

\begin{proof}
By Lemma~\ref{ch2:sparsity} and by using Campbell Theorems~\eqref{ch2: eq:campbell}-\eqref{ch2:eq:campbell2},
\begin{align*}
\mathrm{Var}[ \ell^{(1)}_{n,1}(w; \boldsymbol \beta_0)]= \mathbf{B}_{n,11}(w; \boldsymbol \beta_{0})+\mathbf{C}_{n,11}(w; \boldsymbol \beta_{0}).
\end{align*}
The remainder of the proof follows~\cite{coeurjolly2014variational}. Let $C_i=i+(-1/2,1/2]^d$ be the unit box centered at $i \in \mathbb{Z}^d$ and define $\mathscr{I}_n=\{i \in \mathbb{Z}^d, C_i \cap D_n \neq \emptyset \}$. Set $D_n={\displaystyle \bigcup_{i \in \mathscr{I}_n} C_{i,n}}$, where $C_{i,n}=C_i \cap D_n$. We have
\[
\ell^{(1)}_{n,1}(w;\boldsymbol \beta_0)={\sum_{i \in \mathscr{I}_n} Y_{i,n}}	
\]
where
\[
Y_{i,n}=\!\!\!\!\!\sum_{u \in \mathbf{X} \cap C_{i,n}} \!\!\! w(u)\mathbf{z}_{01}(u) - \int_{ C_{i,n}} w(u)\mathbf{z}_{01}(u)\exp(\boldsymbol\beta_{01}^\top \mathbf z_{01}(u) )\mathrm{d}u.	
\]
For any $n \geq 1$ and any $i \in \mathscr{I}_n$, $Y_{i,n}$ has zero mean, and by condition ($\mathcal C$.\ref{ch2:C:rhok}),
\begin{align}
{\displaystyle \sup_{n \geq 1} \sup_{i \in \mathscr{I}_n} \mathbb{E}(\|Y_{i,n}\|^{2+\delta})} < \infty. \label{ch2:eq:12}
\end{align}

If we combine (\ref{ch2:eq:12}) with conditions  ($\mathcal C$.\ref{ch2:C:Dn})-($\mathcal C$.\ref{ch2:C:BnCn}), we can apply \citet[][Theorem 4]{karaczony2006central}, a central limit theorem for triangular arrays of random fields.

\end{proof}


\bigskip

\begin{proof} We now focus on the proof of Theorem~\ref{ch2:the2}. Since Theorem~\ref{ch2:the2}(i) is proved by Lemma~\ref{ch2:sparsity}, we only need to prove Theorem~\ref{ch2:the2}(ii), which is the asymptotic normality of $\boldsymbol {\hat{\beta}}_1$. As shown in Theorem~\ref{ch2:the1}, there is a root-$(|D_n|/p_n)$ consistent local maximizer $\boldsymbol{\hat{\beta}}$ of $Q_n(w;\boldsymbol \beta)$, and it can be shown that there exists an estimator $\boldsymbol {\hat{\beta}}_1$ in Theorem~\ref{ch2:the1} that is a root-$(|D_n|/p_n)$ consistent local maximizer of $ Q_n \Big(w;({\boldsymbol \beta_1}^\top,\mathbf{0}^\top)^\top \Big)$, which is regarded as a function of  $\boldsymbol {\beta}_1$, and that satisfies
\begin{align*}
\frac {\partial Q_n(w;\boldsymbol {\hat \beta})}{\partial\beta_j}=0 \quad
\mbox { for } j=1,\ldots,s \mbox { and } \boldsymbol{\hat \beta}=( \boldsymbol {\hat{\beta}}_1^\top,\mathbf{0}^ \top)^\top.
\end{align*}
There exists $t\in (0,1)$ and $\boldsymbol{\tilde{\beta}}= \boldsymbol{\hat \beta} + t(\boldsymbol\beta_0-\boldsymbol{\hat \beta})$  such that for $j=1,\cdots,s$
\begin{align}
0
=&\frac {\partial \ell_n{(w;\boldsymbol{\hat \beta})}}{\partial\beta_j}-|D_n|p'_{\lambda_{n,j}}(|\hat \beta_{j}|)\sign(\hat \beta_j) \nonumber\\
=&\frac {\partial \ell_n{(w;\boldsymbol \beta_0)}}{\partial\beta_j}+{\sum_{l=1}^{s} \frac {\partial^2 \ell_n{(w; \boldsymbol{\tilde{\beta}})}}{\partial\beta_j \partial\beta_l}}({\hat \beta_l}-\beta_{0l})-|D_n|p'_{\lambda_{n,j}}(|\hat \beta_{j}|)\sign(\hat \beta_j) \nonumber\\
=&\frac {\partial \ell_n{(w;\boldsymbol \beta_0)}}{\partial\beta_j}+{\sum_{l=1}^{s} \frac {\partial^2 \ell_n{(w; \boldsymbol \beta_0)}}{\partial\beta_j \partial\beta_l}}({\hat \beta_l}-\beta_{0l})+{\sum_{l=1}^{s} \Psi_{n,jl}({\hat \beta_l}-\beta_{0l})} \nonumber \\
&-|D_n|p'_{\lambda_{n,j}}(|\beta_{0j}|)\sign(\beta_{0j})-|D_n|\phi_{n,j}, \label{ch2:eq:0equal}
\end{align}
where 
\begin{align*}
\Psi_{n,jl}=\frac {\partial^2 \ell_n{(w;\boldsymbol{\tilde{\beta}})}}{\partial\beta_j \partial\beta_l}-\frac {\partial^2 \ell_n{(w;\boldsymbol \beta_0)}}{\partial\beta_j \partial\beta_l}
\end{align*}
and $\phi_{n,j}=p'_{\lambda_{n,j}}(|\hat \beta_{j}|)\sign(\hat \beta_j)-p'_{\lambda_{n,j}}(|\beta_{0j}|)\sign(\beta_{0j})$. Since $p'_{\lambda}$ is a Lipschitz function by condition~($\mathcal C$.\ref{ch2:C:plambda}), there exists $\kappa \geq 0$ such that by condition on $a_n$
\begin{align}
\phi_{n,j} & =p'_{\lambda_{n,j}}(|\hat \beta_{j}|)\sign(\hat \beta_j)-p'_{\lambda_{n,j}}(|\beta_{0j}|)\sign(\beta_{0j}) \nonumber\\
 & = \big(p'_{\lambda_{n,j}}(|\hat \beta_{j}|)-p'_{\lambda_{n,j}}(|\beta_{0j}|)\big) \sign(\hat \beta_j) + p'_{\lambda_{n,j}}(|\beta_{0j}|) \big(\sign(\hat \beta_j)-\sign(\beta_{0j})\big) \nonumber \\
& \leq \kappa \big| |\hat \beta_j| - |\beta_{0j}| \big| + 2a_n \nonumber \\
& \leq \kappa | \hat \beta_j - \beta_{0j} | + 2a_n \label{Ch2:T_2}.
\end{align}
We now decompose $\phi_{n,j}$ as $\phi_{n,j}=T_1+T_2$ where
\[
	T_1=\phi_{n,j} \mathbb{I}(|\hat \beta_j- \beta_{0j}| \leq \tilde r_{n,j}) 
	\quad \mbox{ and }\quad 
	T_2=\phi_{n,j} \mathbb{I}(|\hat \beta_j-\beta_{0j}| > \tilde r_{n,j})
\]
and where $\tilde r_{n,j}$ is the sequence defined in the condition ($\mathcal C$.\ref{ch2:C:plambda}). Under this condition, the following Taylor expansion can be derived for the term $T_1$: there exists $t\in (0,1)$ and $\check\beta_j= \hat\beta_j+t(\beta_{0j}-\hat\beta_j)$ such that 
\begin{align*}
 T_1&  =p''_{\lambda_{n,j}}(|\beta_{0j}|)(\hat \beta_j- \beta_{0j}) \mathbb{I}(|\hat \beta_j- \beta_{0j}| \leq \tilde r_{n,j}) \\
 &\quad + \frac12 (\hat\beta_{j} - \beta_{0j})^2 p'''_{\lambda_{n,j}}(|\tilde{\beta}_{j}|) \mathrm{sign}(\check \beta_j) \mathbb{I}(|\hat \beta_j- \beta_{0j}| \leq \tilde r_{n,j}) \\
 &=p''_{\lambda_{n,j}}(|\beta_{0j}|)(\hat \beta_j- \beta_{0j}) \mathbb{I}(|\hat \beta_j- \beta_{0j}| \leq \tilde r_{n,j}) + O_{\mathrm P}(p_n/|D_n|)
\end{align*}
where the latter equation ensues from Theorem~\ref{ch2:the1} and condition~($\mathcal C$.\ref{ch2:C:plambda}). Again,
from Theorem~\ref{ch2:the1},  $\mathbb{I}(|\hat \beta_j- \beta_{0j}| \leq \tilde r_{n,j}) \xrightarrow{L^1} 1$  which implies that $\mathbb{I}(|\hat \beta_j- \beta_{0j}| \leq \tilde r_{n,j}) \xrightarrow{\mathrm{P}} 1$,  so $ T_1=p''_{\lambda_{n,j}}(|{\beta}_{0j}|)(\hat \beta_j- \beta_{0j}) \big(1+o_{\mathrm{P}}(1)\big)+  O_{\mathrm P}(p_n/|D_n|)$.

Regarding the term $T_2$, we have by \eqref{Ch2:T_2}
\begin{align*}
T_2 &\leq \{\kappa | \hat \beta_j - \beta_{0j} | + 2a_n \label{T_2}\} \; \mathbb{I}(|\hat \beta_j-\beta_{0j}| > \tilde r_{n,j}) \\
&= \kappa | \hat \beta_j - \beta_{0j} | \; \mathbb{I}(|\hat \beta_j-\beta_{0j}| > \tilde r_{n,j}) + o(|D_n|^{-1/2}).	
\end{align*}
We want to prove that $T_2 = o_{\mathrm{P}}(|D_n|^{-1/2})$. Define $S_n=|\hat \beta_j-\beta_{0j}| \;\mathbb{I}(|\hat \beta_j-\beta_{0j}| > \tilde r_{n,j})$ and $T_n=\mathbb{I}(S_n> \delta|D_n|^{-1/2})$ for some $\delta>0$. We claim that the result is proved if we prove that $\mathbb{E}T_n \to 0$ for any $\delta>0$.  Condition~($\mathcal C$.\ref{ch2:C:plambda}) implies in particular that for $n$ large enough, $\tilde r_{n,j} > \sqrt{p_n/|D_n|} > \sqrt{1/|D_n|}$. Using this, it can be checked that the binary random variable $T_n$ reduces to $T_n=\mathbb{I}(|\hat \beta_j-\beta_{0j}| > \tilde r_{n,j})\stackrel{L^1}{\to}0$ as $n\to \infty$.

Then, we  deduce that 
\begin{equation} \label{eq:phinj}
	\phi_{n,j} = p''_{\lambda_{n,j}}(|{\beta}_{0j}|)(\hat \beta_j- \beta_{0j}) \big(1+o_{\mathrm{P}}(1)\big)+ O_{\mathrm P}(p_n/|D_n|) + o_{\mathrm P}(|D_n|^{-1/2}).
\end{equation}

Let $\ell^{(1)}_{n,1}(w;\boldsymbol \beta_{0})$ (resp. $\ell^{(2)}_{n,1}(w;\boldsymbol \beta_{0})$) be the first $s$ components (resp. $s \times s$ top-left corner) of $\ell^{(1)}_{n}(w;\boldsymbol \beta_{0})$ (resp. $\ell^{(2)}_{n}(w;\boldsymbol \beta_{0})$). Let also $\boldsymbol \Psi_n$ be the $s \times s$ matrix containing $\Psi_{n,jl}, j,l=1,\ldots,s$. Finally, let the vector $\mathbf{p}'_n$, the vector $\boldsymbol \phi_n$ and the $s \times s$ matrix $\mathbf{M}_n$ be defined by
\begin{align*}
\mathbf{p}'_n&=\{p'_{\lambda_{n,1}}(|\beta_{01}|)\sign(\beta_{01}),\ldots,p'_{\lambda_{n,s}}(|\beta_{0s}|)\sign(\beta_{0s})\}^\top, \\
\boldsymbol \phi_n&=\{\phi_{n,1},\ldots,\phi_{n,s}\}^\top, \mbox{ and}\\
\mathbf{M}_n&=\{ \mathbf{B}_{n,11}(w; \boldsymbol \beta_{0})+\mathbf{C}_{n,11}(w; \boldsymbol \beta_{0})\}^{-1/2}.
\end{align*}
We rewrite both sides of~\eqref{ch2:eq:0equal} as
\begin{equation}
\ell^{(1)}_{n,1}(w;\boldsymbol \beta_{0})+\ell^{(2)}_{n,1}(w;\boldsymbol \beta_{0})(\boldsymbol{\hat \beta}_1-\boldsymbol \beta_{01})+ \boldsymbol \Psi_n (\boldsymbol{\hat \beta}_1-\boldsymbol \beta_{01}) -|D_n| \mathbf{p}'_n-|D_n| \boldsymbol \phi_n   =0.\label{eq:0vec}
\end{equation}
By definition of $\boldsymbol \Pi_n$ given by (\ref{ch2:eq:pi}) and from~\eqref{eq:phinj}, we obtain $\boldsymbol \phi_n=\boldsymbol \Pi_n (\boldsymbol{\hat \beta}_1-\boldsymbol \beta_{01})\big(1+o_{\mathrm{P}}(1)\big)+ O_{\mathrm P}({p_n/|D_n|}) +  o_{\mathrm P}(|D_n|^{-1/2})$. Using this, we deduce, by premultiplying both sides of~\eqref{eq:0vec} by $\mathbf M_n$, that
\begin{align*}
\mathbf{M}_n \ell^{(1)}_{n,1}(w;\boldsymbol \beta_{0})-&\mathbf{M}_n \big(\mathbf{A}_{n,11}(w;\boldsymbol \beta_{0})+ |D_n| \boldsymbol \Pi_n\big)(\boldsymbol{\hat \beta}_1-\boldsymbol \beta_{01}) \\
& =O(|D_n| \, \|\mathbf{M}_n \mathbf{p}'_n \|) + o_{\mathrm{P}}(|D_n| \, \|\mathbf{M}_n\boldsymbol \Pi_n (\boldsymbol{\hat \beta}_1-\boldsymbol \beta_{01}) \|)  \\
& \quad+ O_{\mathrm{P}} (\|\mathbf M_n \| p_n) +  o_{\mathrm P}(\|\mathbf M_n \| |D_n|^{1/2}) \\
& \quad+ O_{\mathrm{P}} (\|\mathbf M_n \boldsymbol \Psi_n (\boldsymbol{\hat \beta}_1-\boldsymbol \beta_{01}) \|) .
\end{align*}

Now, $\|\mathbf M_n\|=O(1/ \sqrt{|D_n|})$  by condition ($\mathcal C$.\ref{ch2:C:BnCn}), $\| \boldsymbol \Psi_n\|=O_{\mathrm{P}}(\sqrt{p_n |D_n|})$  by conditions ($\mathcal C$.\ref{ch2:C:Theta})-($\mathcal C$.\ref{ch2:C:cov}) and Theorem~\ref{ch2:the1}, and $\|\boldsymbol{\hat \beta}_1-\boldsymbol \beta_{01}\|=O_{\mathrm{P}}(\sqrt{p_n/|D_n|})$ by Theorem~\ref{ch2:the1} and Theorem~\ref{ch2:the2}(i). Finally, since by assumptions that \linebreak $a_n \sqrt{|D_n|}\to 0$ and $c_n \sqrt{p_n} \to 0$ as $n \to \infty$, we deduce that
\begin{align*}
|D_n| \,\|\mathbf{M}_n \mathbf{p}'_n\| &= O(a_n \sqrt{D_n}) =o(1),\\
|D_n| \, \|\mathbf{M}_n\boldsymbol \Pi_n (\boldsymbol{\hat \beta}_1-\boldsymbol \beta_{01}) \| &= O_\mathrm{P}\left(\sqrt{|D_n|} c_n \sqrt{\frac{p_n}{|D_n|}} \right) =o_\mathrm{P}(1)  ,\\
 \|\mathbf M_n \| \; \sqrt{|D_n|} &=  O(1),\\
\|\mathbf M_n \| \; p_n &=  O\left( \sqrt{\frac{p_n^2}{|D_n|}}\right)=o(1),\\
\|\mathbf{M}_n \boldsymbol \Psi_n (\boldsymbol{\hat \beta}_1-\boldsymbol \beta_{01})\| &=O_{\mathrm{P}}\left( \sqrt{\frac{p_n^2}{|D_n|}}\right)=o_{\mathrm{P}}(1).
\end{align*}
The last two lines are obtained from ($\mathcal C$.\ref{ch2:C9}). Therefore, we have that
\begin{align*}
\mathbf{M}_n \ell^{(1)}_{n,1}(w;\boldsymbol \beta_{0})-\mathbf{M}_n \big(\mathbf{A}_{n,11}(w;\boldsymbol \beta_{0})+ |D_n| \boldsymbol \Pi_n\big)(\boldsymbol{\hat \beta}_1-\boldsymbol \beta_{01}) =o_{\mathrm{P}}(1).
\end{align*}
By (\ref{eq:normal}) in Lemma~\ref{ch2:clt} and by Slutsky's Theorem, we deduce that
\begin{align*}
\{ \mathbf{B}_{n,11}(w; \boldsymbol \beta_{0})+\mathbf{C}_{n,11}(w; \boldsymbol \beta_{0})\}^{-1/2} \times \\
\{\mathbf{A}_{n,11}(w;\boldsymbol \beta_{0})+|D_n| \boldsymbol \Pi_n\}(\boldsymbol{\hat \beta}_1-\boldsymbol \beta_{01})&\xrightarrow{d} \mathcal{N}(0,\mathbf{I}_{s})
\end{align*}
as $n \to \infty$, which can be rewritten, in particular under ($\mathcal C$.\ref{ch2:C:An}), as 
\[
|D_n|^{1/2}\boldsymbol \Sigma_n(w;\boldsymbol \beta_{0})^{-1/2}(\boldsymbol{\hat \beta}_1-\boldsymbol \beta_{01})\xrightarrow{d}\mathcal{N}(0,\mathbf{I}_{s})	
\]
where $\mathbf \Sigma_n(w,\boldsymbol \beta_{0})$ is given by~\eqref{ch2:eq:Sigman}.
\end{proof}



\bibhang=1.7pc
\bibsep=2pt
\fontsize{9}{14pt plus.8pt minus .6pt}\selectfont
\renewcommand\bibname{\large \bf References}
\expandafter\ifx\csname
natexlab\endcsname\relax\def\natexlab#1{#1}\fi
\expandafter\ifx\csname url\endcsname\relax
  \def\url#1{\texttt{#1}}\fi
\expandafter\ifx\csname urlprefix\endcsname\relax\def\urlprefix{URL}\fi

\vskip .65cm
\noindent
Department of Mathematical Sciences, Aalborg University, Denmark
\vskip 1pt
\noindent
E-mail: achmad@math.aau.dk
\vskip 6pt

\noindent
Department of Mathematics, Universit\'e du Qu\'ebec \`a Montr\'eal (UQAM), Canada and\\
Department of Probability and Statistics, Universit\'e Grenoble Alpes, France
\vskip 1pt
\noindent
E-mail: coeurjolly.jean-francois@uqam.ca
\vskip 6pt

\noindent
Department of Probability and Statistics, Universit\'e Grenoble Alpes, France
\vskip 1pt
\noindent
E-mail: frederique.letue@univ-grenoble-alpes.fr 


\end{document}